\documentclass[twoside,11pt]{article}

\usepackage{jmlr2e}

\usepackage{subfig}
\usepackage{amsmath}
\usepackage{url}
\usepackage{comment}
\usepackage[vlined,ruled,boxed]{algorithm2e}
\usepackage{booktabs}
\usepackage{xcolor}

\newcommand{\indep}{\mathrel{\text{\scalebox{1.07}{$\perp\mkern-10mu\perp$}}}}

\newcommand{\xhdr}[1]{\vspace{1mm}{\bf #1}}

\newcommand{\invalidiv}{Invalid-IV }
\newcommand{\validiv}{Valid-IV }
\newcommand{\excl}{Exclusion }
\newcommand{\air}{As-if-random }

\jmlrheading{}{}{}{}{}{}{}

\ShortHeadings{NPS test for instrumental variables}{Sharma}
\firstpageno{1}

\begin{document}

\title{Necessary and Probably Sufficient Test for Finding Valid Instrumental Variables}


\author{\name Amit Sharma \email amshar@microsoft.com\\
       \addr Microsoft Research India\\
       Vigyan, Bangalore 560008
        }
\editor{In-progress draft}

\maketitle

\begin{abstract}
    Can instrumental variables be found from data? While instrumental variable (IV) methods are widely used to identify causal effect, testing their validity from observed data remains a challenge. This is because validity of an IV depends on two assumptions, \emph{exclusion} and \emph{as-if-random}, that are largely believed to be untestable from data.  
In this paper, we show that under certain conditions,  testing for instrumental variables is possible. We build upon prior work on necessary tests to derive a test that characterizes the odds of being a valid instrument, thus yielding the name ``necessary and \emph{probably} sufficient''. 
The test works by defining the class of invalid-IV and valid-IV causal models as Bayesian generative models and comparing their marginal likelihood based on observed data. When all variables are discrete, we also provide a method to efficiently compute these marginal likelihoods.
\\
We evaluate the test on an extensive set of simulations for binary data, inspired by an open problem for IV testing proposed in past work. We find that the test is most powerful when an instrument follows monotonicity---effect on treatment is either non-decreasing or non-increasing---and has moderate-to-weak strength; incidentally, such instruments are  commonly used in observational studies. Among as-if-random and exclusion, it detects exclusion violations with higher power. 
Applying the test to IVs from two seminal studies on instrumental variables and five recent studies from the American Economic Review shows that many of the instruments may be flawed, at least when all variables are discretized. 
The proposed test opens the possibility of  data-driven validation and search for instrumental variables. 
\end{abstract}




\section{INTRODUCTION}
The method of \textit{instrumental variables} is one of the most popular ways to estimate causal effects from observational data in the social and biomedical sciences. The key idea is to find subsets of the data that resemble a randomized experiment, and use those subsets to estimate causal effect. For example, instrumental variables have been used in economics to study the effect of policies such as military conscription and compulsory schooling  on future earnings \citep{angrist1991,angrist2008}, and in epidemiology (under the name \textit{Mendelian} randomization) to study the effect of risk factors on disease outcomes \citep{lawlor2008}. 

Inspite of their popularity, basic questions about the design, analysis and evaluation of instrumental variable (IV) studies remain elusive. 
In the design phase, it is unclear how to find a suitable instrumental variable.  Even with access to bigger and more granular datasets, finding an instrument requires ingenuity and a laborious manual search, thereby restricting most IV studies to instruments derived from a small set of events such as the weather, lotteries or sudden shocks~\citep{dunning2012natural}. 
In the analysis phase, arguments are proposed to justify selection of an instrument, but it is hard to ascertain the extent to, or for which populations,  the instrument is likely to be valid. 
Finally, in the evaluation phase, it is unclear how to evaluate the estimates produced by multiple IV studies due to the absence of any objective criteria of comparing relative validity of instruments, even on the same dataset.

Specifically, consider the canonical causal inference problem shown in  Figure~\ref{fig:std-iv-model}a. The goal is to estimate the effect of a variable $X$ on another variable $Y$ 
based on observed data, where $X$ is commonly referred to as the \emph{treatment} and $Y$ as the \emph{outcome}.
However, there are unobserved (and possibly unknown) common causes for X and Y that confound observed association between X and Y, making the isolation of X's effect on Y a non-trivial problem.
Unlike methods such as stratification or matching that condition on all observed common causes~\citep{stuart2010matching}, the instrumental variable method relies on finding an additional variable $Z$ that acts as an \textit{instrument} to modify the distribution of $X$, as shown by the arrow $Z \rightarrow X$ in Figure~\ref{fig:std-iv-model}a.  The advantage is that we do not need to assume that all confounding common causes are observed to estimate the causal effect. 
To be a valid instrument, however,  $Z$ should satisfy three conditions \citep{angrist2008}. First, $Z$ should have a substantial effect on $X$. That is, $Z$ causes $X$ \textit{(Relevance)}. Second, $Z$ should not cause $Y$ directly (\textit{Exclusion}); the only association between $Z$ and $Y$ should be through $X$. Third, $Z$ should be independent of all the common causes $U$ of $X$ and $Y$ (\textit{As-if-random}). The latter two conditions are shown in the graphical model in Figure~\ref{fig:std-iv-model}b. These conditions can also be expressed as conditional independence constraints: exclusion and as-if-random conditions imply $Z \indep Y | X, U $ and $Z\indep U$ respectively.

However, the Achilles' heel of any instrumental variable analysis is that these core conditions are never tested systematically. Except for relevance (which  can be tested by measuring  observed correlation between $Z$ and $X$), the other two conditions depend on unobserved variables $U$ and thus are harder to check. 
Although necessary tests do exist that can weed out bad instruments~\citep{pearl1995-iv,bonet2001}, in practice exclusion and as-if-random as considered as \textit{assumptions} and often defended with qualitative domain knowledge.
This can be problematic because the entire validity of the IV estimate depends on the exclusion and as-if-random conditions. 

In this paper, therefore we propose a test for validating instrumental variables that can be used to find, evaluate and compare potential instruments for their validity. Although instruments are untestable in general~\citep{morgan2014counterfactuals,dunning2012natural}, we find that in many cases it is possible to distinguish between invalid and valid instruments. 
To do so, the proposed test applies the principles of Bayesian model comparison to causal models and estimates marginal likelihood of an valid instrument given the observed data. Comparing this to the corresponding marginal likelihood for an invalid instrument provides a metric for evaluating the validity of an instrument. The intuition is that if the instrument is valid, then causal models with an instrument as in Figure~\ref{fig:std-iv-model}a should be able to generate observed data with higher likelihood than all other causal models.  
Specifically, let \textit{Valid-IV} refer to the class of all causal models that yield a valid instrument and \textit{Invalid-IV} to the class of causal models that yield an invalid instrument. Given an observed data distribution $P(X, Y, Z)$, the proposed method computes the ratio of marginal likelihoods for Valid-IV and Invalid-IV models. Whenever this marginal likelihood ratio is above a pre-determined acceptance threshold, we can conclude that the instrument is likely to be valid. To distinguish this probabilistic notion from deterministic sufficiency---conditions that would determine in absolute whether an instrument is valid or not---we call an instrument that passes the marginal likelihood ratio test as \textit{probably sufficient}.

\begin{figure}%
	\centering
	\subfloat[Valid-IV causal model]{\includegraphics[scale=0.5]{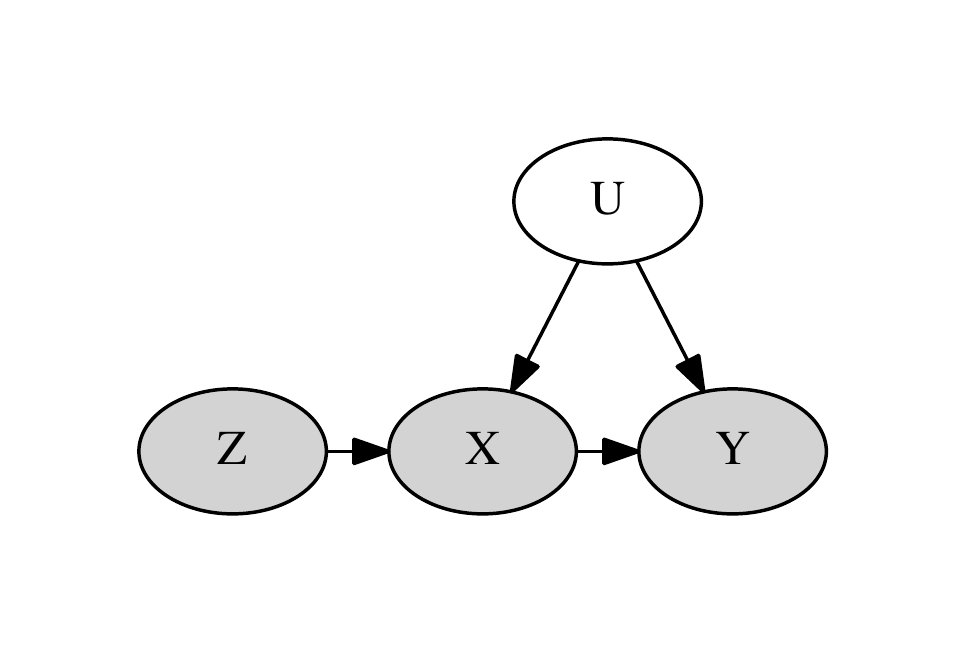}}%
	\qquad
	\subfloat[Invalid-IV causal model]{\includegraphics[scale=0.5]{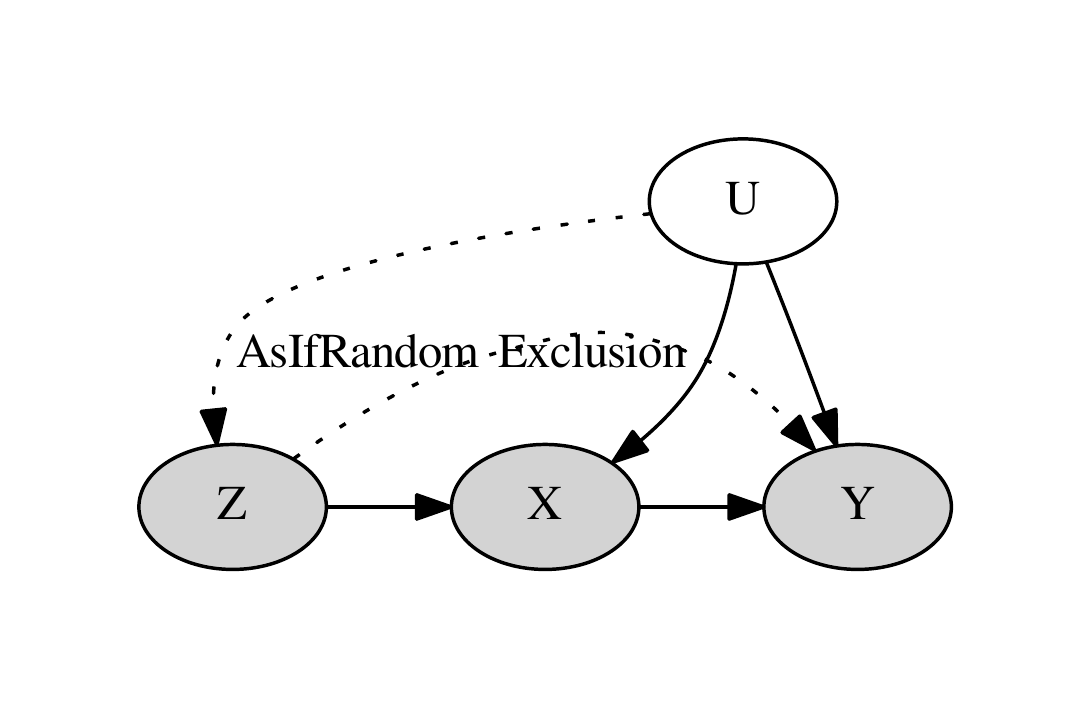}}%
	\caption{Standard instrumental variable causal model and common violations that lead to an invalid-IV model. Exclusion condition is violated when the instrumental variable $Z$ directly affects the outcome $Y$. As-if-random condition is violated when unobserved confounders $U$ also affect the instrumental variable $Z$. }
	\label{fig:std-iv-model}
\end{figure}

Combining the above approach with  necessary tests proposed in past work leads to a \textit{Necessary and Probably Sufficient (NPS)} test for instrumental variables.  
The combined NPS test proceeds as follows. 
If the observed data does not satisfy the necessary conditions, then it is declared invalid. If it does, then we proceed to estimate the marginal likelihood ratio over Valid-IV and Invalid-IV models. When all variables are discrete, we provide a general implementation of this test that makes no assumptions about the nature of functional relationships between the treatment, outcome and instrument.


Finally, any statistical test is only as good as the decisions it helps to support.   Among the two IV assumptions, simulations show that the NPS test is more effective at detecting violations of the exclusion restriction.
In particular, for certain instrumental variable designs where we restrict the direction of causal effects, such as stipulating that the instrument can only have a monotonic effect on treatment and possibly also the outcome,  the NPS test is able to detect invalid instruments with high power. We also find that 
the proposed NPS test is most effective for validating instruments having low correlation with the treatment $X$. 
Incidentally, most of the instruments used in observational studies in the social and biomedical sciences have weak to moderate strength, well-suited to the NPS test.   To demonstrate the test's usefulness in practice, we first consider an open problem proposed by \cite{palmer2011ivbounds} for validating an instrumental variable and show that the NPS test is able to identify valid instruments in that setting. Second,  we apply the test to datasets from two seminal studies and five recent papers on instrumental variables from the American Economic Review, a premier economics journal. In many cases, we find that instrumental assumptions used in the corresponding papers were possibly flawed, at least when variables are binarized. 
We provide an R package \textit{ivtest} that provides an implementation of the NPS test.

\section{BACKGROUND: TESTABILITY OF AN INSTRUMENTAL VARIABLE}
\label{sec:related-work}
Since sufficient conditions for validity of an instrument ($Z\indep U$ and $Z \indep Y | X, U$) depend on an unobserved variable $U$, the validity of an instrumental variable is largely believed to be untestable from observational data alone \citep{morgan2014counterfactuals}. \cite{pearl1995-iv}, however,  discovered that the specific causal graph structure in Figure~\ref{fig:std-iv-model} imposes constraints on the observed probability distribution over $Z$, $X$ and $Y$. These constraints can be used to derive a necessary test for IV conditions \citep{pearl1995-iv,bonet2001}. Such a test weeds out bad instruments, but is inconclusive whenever an instrument passes the test. 
Sufficient tests exist, but require prohibitive assumptions such as knowing another valid instrumental variable as in the  Durbin-Wu-Hausman test \citep{nakamura1981}, or stipulating that confounders have no effect on the outcome.  
Due to these shortcomings, the dominant method to evaluate IV estimates involves a sensitivity analysis~\citep{small2007-ivsensitivity}, where we estimate the impact on the causal estimate when assumptions are violated with varying severity. 
We review prior tests and sensitivity analysis for instrumental variables below.

\subsection{Tests for an instrumental variable}
A classic test for instrumental variables is the Durbin-Wu-Hausman test~\citep{nakamura1981}. Given a subset of valid instrumental variables, it  can  identify whether other potential candidates are also valid instrumental variables. However, it provides no guidance on how to find the initial set of valid instrumental variables.  

Without having an initial set of valid instrumental variables, an intuitive idea is to test whether the instrument is independent of the outcome whenever the treatment is constant ($Z \indep Y|X$)~\citep{balke1993nonparametric}. This can be seen as an approximation of the exclusion restriction using only observed data. It is sufficient under the assumption that either all common causes $U$ are constant throughout the data measurement period or $U$ trivially are not common causes---each of them can affect either $X$ or $Y$, but not both\footnote{Here we make the standard assumption of \textit{faithfulness}. That is, observed conditional independence between variables implies causal independence. }.  
However, under any non-trivial common cause $U$, it will be implausible that $Z$ and $Y$ are independent, because conditioning on $X$ induces a correlation between $Z$ and $U$ (and thus $Y$).
Further, because we ignore the contribution of unobserved confounders, this test is not a necessary condition for a valid instrument.

\label{sec:rw-graphical-tests}
More admissive tests can be obtained by considering the restriction on probability distribution for $(Z, X, Y)$ imposed by a valid IV model. Consider the causal IV model from Figure~\ref{fig:std-iv-model}a. In structural equations, the model can be equivalently expressed as: 
\begin{align} \label{eqn:structural-iv}
y &= f(x, u); \text{  \ \ }x = g(z, u)
\end{align}
where $f$ and $g$ are arbitrary deterministic functions and $U$ represents arbitrary, unobserved random variables that are independent of $Z$. 
Using this framework, Pearl derived conditions that any observed data generated from a valid instrumental variable model must satisfy~\citep{pearl1995-iv}. For binary variables $Z$, $X$ and $Y$, the Pearl's IV test can be written as the following set of \textit{instrumental inequalities}.
\begin{align}
P(Y=0, X=0|Z=0) + P(Y=1,X=0|Z=1) &\leq 1 \nonumber \\
P(Y=0, X=1|Z=0) + P(Y=1,X=1|Z=1) &\leq 1 \nonumber \\
P(Y=1, X=0|Z=0) + P(Y=0,X=0|Z=1) &\leq 1 \nonumber \\
P(Y=1, X=1|Z=0) + P(Y=0,X=1|Z=1) &\leq 1 
\end{align}

Typically, researchers make an additional assumption that helps to derive a point estimate for the  Local Average Treatment Effect (LATE). This assumption, called monotonicity \citep{angrist1994identification}, 
precludes any \textit{defiers} to treatment in the population \citep{angrist2008}. That is, we assume that $g(z_1, u) \geq g(z_2, u)$ whenever $z_1 \geq z_2$.  Under these conditions, Pearl showed that we can obtain tighter inequalities.  For binary variables $Z$, $X$ and $Y$, the instrumental inequalities become:
\begin{align} \label{eqn:binary-monotonicity}
P(Y=y, X=1|Z=1) & \geq P(Y=y, X=1|Z=0)  &\forall y \in \{0,1\} \\
P(Y=y, X=0|Z=0) & \geq P(Y=y, X=0|Z=1) &\forall y \in \{0,1\}
\end{align}

Whenever any of these inequalities are violated, it implies that one or more of the IV assumptions---exclusion, as-if-random or monotonicity---are violated. Based on these instrumental inequalities, different hypothesis tests have been proposed to account for sampling variability in observing the true conditional distributions. For example, a null hypothesis tests based on the chi-squared statistic \citep{ramsahai2011} or the Kolmogorov-Smirnov test statistic \citep{kitagawa2015-ivtest} can be applied.   

Moreover, when $X$, $Y$ and $Z$ are binary, this test is not only necessary, it is the strongest necessary test possible \citep{bonet2001,kitagawa2015-ivtest}. In other words, if an observed data distribution satisfies the test, then there does exist at least one valid-IV  causal model that could have generated the data; we call this the \textit{existence} property. However, the test does not satisfy the existence property when all variables are not binary, allowing probability distributions that cannot be generated by any valid-IV model. To rectify this, Bonet proposed a more general version of the test that ensures the existence property for any discrete-valued $X$, $Y$ and $Z$. \citep{bonet2001}. We refer to this version of the test as the Pearl-Bonet necessary and existence test for instrumental variables, or simply the \textit{Pearl-Bonet test}.

While Bonet presented theoretical properties of the test for discrete variables, implementing the test in practice is non-trivial because it involves testing membership of a convex polytope in high-dimensional space. Further, the test does not support the monotonicity assumption, a popular assumption in instrumental variable studies. In this paper, therefore,  we extend Bonet's work by incorporating monotonicity and present a practical method for testing IVs when variables can have arbitrary number of discrete levels. 


\begin{figure}[ht]
    \centering
	\subfloat[Data distributions generated by Valid-IV and Invalid-IV models]{\includegraphics[scale=0.25]{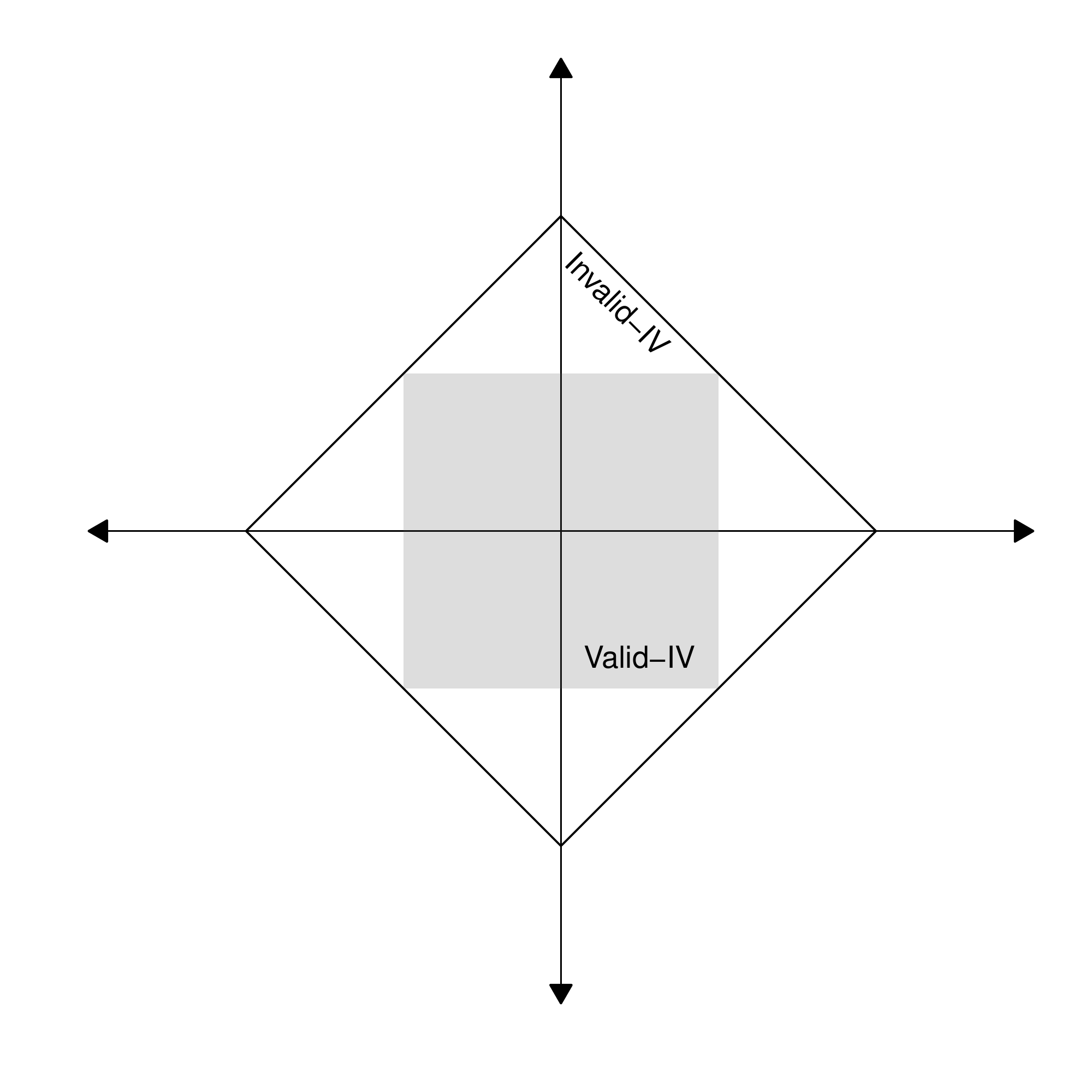}}%
	\qquad
	\subfloat[Data distributions that pass or fail Pearl's necessary test]{\includegraphics[scale=0.42]{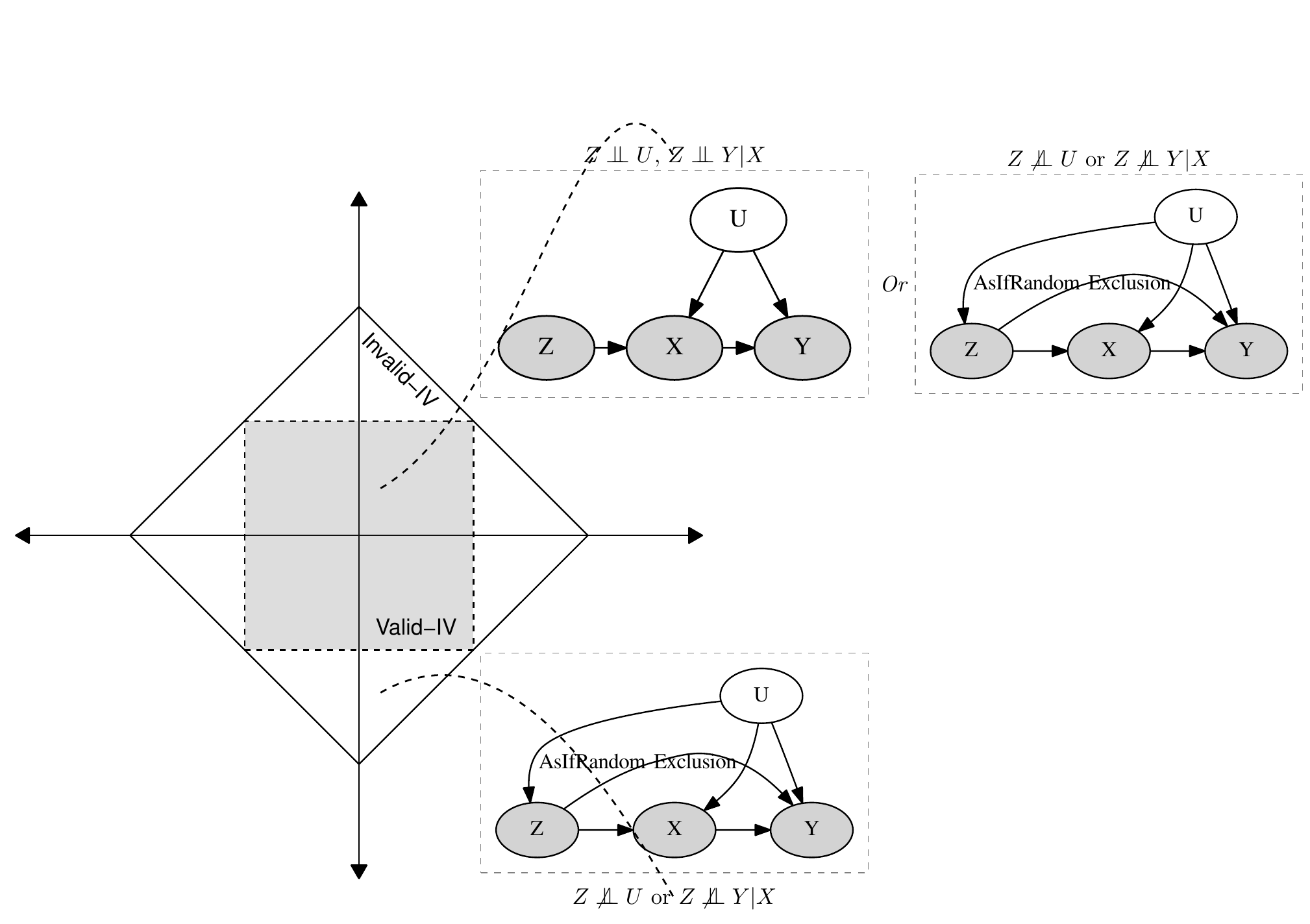}}%
	\qquad
    \caption{\textbf{(a)} A 2D schematic of the probability space $P(X,Y|Z)$. 
	Each point on the schematic plot is a conditional probability distribution over $X$, $Y$ given $Z$. The two squares show polytope boundaries for all distributions generated by invalid-IV and valid-IV models. 
Note that data distributions generated by valid-IV are a subset of  the data distributions generated by invalid-IV. \textbf{(b)} The right panel shows the test boundary for Pearl-Bonet necessary test in dotted lines, which coincides with the Valid-IV boundary for binary $X$. Thus, if a data distribution $P(X,Y|Z)$ fails the test, it has to be generated from a invalid-IV model, as shown in the bottom right inset. Still, Pearl-Bonet test is not sufficient because a data distribution $P(X, Y|Z)$ that passes the test--- and thus lies inside the valid-IV polytope---may be generated by either a valid-IV or invalid-IV model, as shown in the top right inset.}
	\label{fig:nectest-iv-schematic}
\end{figure}

\subsection{Sensitivity analysis for instrumental variables}
Still, the above tests can refute an invalid-IV model, but are unable to \textit{verify} a valid-IV model \citep{kitagawa2015-ivtest}. That is, even when an observed data distribution passes the necessary test, it does not exclude the possibility that data was generated by an invalid-IV model. As an example, consider the following data distribution presented by \cite{palmer2011ivbounds} which passes the Pearl-Bonet test, but violates the exclusion restriction because the instrument $Z$ causes $Y$ directly. 
\begin{align}
    \centering
    Z \sim & Bern(0.5) \nonumber \\
    U \sim & Bern(0.5) \nonumber\\
    X \sim & Bern(p_X); p_X = 0.05 + 0.1Z+0.1U \nonumber\\
    Y \sim & Bern(p_Y); p_2 = 0.1 + 0.05Z + 0.05X + 0.1U 
\end{align}

More generally, let us look at Figure~\ref{fig:nectest-iv-schematic} that shows the relationship between an observed data distribution and Valid-IV or Invalid-IV models. Each point in Figure~\ref{fig:nectest-iv-schematic}a represents a probability distribution over $X$, $Y$ and $Z$.\footnote{The axes represent the space of conditional probabilities $P(X, Y|Z)$. We use the fact that any observed probability distribution $\mathcal{P}$ over $X$, $Y$ and $Z$ can be specified by a set of conditional probabilities of the form $P(X=x, Y=y|Z=z)$. For example, for binary variables, this would be a set of eight conditional probabilities  \citep{bonet2001}. The corresponding 8-dimensional real vector would be:
\begin{align}
F (\mathcal{P}) = \big( &P(X=0,Y=0|Z=0), P(X=0,Y=1|Z=0), P(X=1,Y=0|Z=0), P(X=1,Y=1|Z=0), \nonumber \\
				& P(X=0,Y=0|Z=1), P(X=0,Y=1|Z=1), P(X=1,Y=0|Z=1), P(X=1,Y=1|Z=1) \big)
\end{align}
The 2-D squares represented in Figures~\ref{fig:nectest-iv-schematic}a,b are actually polytopes in this multi-dimensional space. The extreme points for $F(\mathcal{P})$, or equivalently for invalid-IV models  are characterized by $P(X=x,Y=y|Z=z)=1 \ \forall z$. The boundary shown for NPS test in Figure~\ref{fig:nectest-iv-schematic}c is however, an oversimplification. The set of conditional distributions (or equivalently, instruments) that can be validated by the NPS test is unknown and most likely will constitute many regions in the probability space, instead of a single bounded region as shown.} The two squares bound the probability distributions that can be generated by any Valid-IV or Invalid-IV model. As can be seen from the figure, probability distributions generatable from  the class of Valid-IV models are a strict subset of the distributions generatable by the class of invalid-IV models. This implies that even if a statistical test can accurately identify the boundary for valid-IV models, as in Figure~\ref{fig:nectest-iv-schematic}b, we can never be sure whether the probability distribution was actually generated by a valid-IV or invalid-IV model.  

As a partial remedy, researchers employ a sensitivity analysis to check the brittleness of a causal estimate when required assumptions are violated. For instance, one may progressively increase the magnitude of violation of the exclusion restriction, and observe when a resultant causal estimate flips in its direction of stated effect~\citep{small2007-ivsensitivity}.
A Bayesian framework presents an intuitive way to conceptualize sensitivity analysis. The idea is to first estimate a causal effect assuming that the data is generated from a valid-IV causal model. One can then change parameters of the causal model to violate key assumptions and observe how fast the causal estimate changes. 
As an example, \cite{kraay2010} proposes a Bayesian analysis to check the sensitivity to assumptions for a linear IV model. This analysis shows that even moderate uncertainty in the prior for the exclusion restriction lead to considerable loss of precision in estimating causal effects.

\subsection{Developing a sufficient test for instrumental variables}
In an ideal world, we would like to reduce uncertainty about assumptions as much as possible and determine precisely whether observed data was generated from a valid-IV model or not. However, as Figure~\ref{fig:nectest-iv-schematic} indicates, establishing \textit{sufficiency} for a validity test is non-trivial. In particular, the usual method of comparing the maximum data likelihood of the two classes of IV models, Valid-IV or Invalid-IV, provides us no information.
This is because \invalidiv class of models (as shown in Figures~\ref{fig:std-iv-model}b and \ref{fig:nectest-iv-schematic}b) is more general than the \validiv class and thus is always as likely (or more) to generate the observed data.

Instead of comparing \emph{maximum} likelihoods of model classes, we turn to estimating likelihoods of individual causal models from \invalidiv and \validiv classes. The intuition is that while the \invalidiv class may always have a causal model that matches likelihood of the \validiv class for a valid instrument, there will be many other \invalidiv models that provide a lower likelihood for the data. By generating models with different violations of the \excl and \air conditions, we can estimate the data likelihood over individual models in the \invalidiv class. Averaging over all models in the \validiv and \invalidiv classes, we expect marginal likelihood to be higher for the \validiv class for data generated from a \validiv model.  The idea of comparing different models from Valid-IV and Invalid-IV classes is similar to sensitivity analysis, except that we are interested in the likelihood of data rather than resultant causal estimates. 

Unlike necessary tests \citep{ramsahai2011,kitagawa2015-ivtest} that refute a null hypothesis that observed data was generated from a valid-IV model, probable sufficiency requires estimating the relative probability of valid-IV and invalid-IV models given observed data. When the relative probability---formally, \emph{marginal likelihood}---is high, the instrument is likely to be valid. Conversely, when it is low, the instrument is likely to be invalid.
Based on this motivation, we now provide a definition for probable sufficiency.

\xhdr{Probable Sufficiency for Instrumental Variables: } If an observed data distribution passes the Pearl-Bonet necessary test, how likely is it that it was generated from a valid-IV model compared to an invalid-IV model?

Intuitively, we wish to find out how often does the Pearl-Bonet test accept an Invalid-IV model. That is, how often does an observed distribution that was generated by an invalid-IV model pass the necessary test? Once we know that, we can compute the probability that a given observed distribution was generated by a valid-IV model, based on the result of Pearl-Bonet test.  

Combined, the Pearl-Bonet test and our probable sufficiency test provide a framework for testing instrumental variables, which we call the \textit{Necessary and Probably Sufficient (NPS)} test for instrumental variables.  Any valid instrument needs to pass Pearl-Bonet test. Therefore, NPS test provides necessity: any instrument that fails instrumentality inequalities is not a valid instrument. Further, NPS test provides sufficiency: any instrument that satisfies Pearl's instrumental inequalities and passes the probable sufficiency test can be accepted as a valid instrument. That said, NPS test will be inconclusive for some instruments: those that satisfy Pearl's inequalities but the marginal likelihood ratio remains close to 1. Figure~\ref{fig:nectest-iv-schematic} shows these possibilities. As shown by the dark grey box in the center, NPS test validates a subset of all possible valid-IV models.



In the next two sections, we describe the NPS test formally. Section~\ref{sec:nps-test-theory} presents a general \textit{Validity Ratio} statistic that can be used to compare Valid-IV and Invalid-IV models. We do so by introducing a probabilistic generative meta-model that formalizes the connection between IV assumptions, causal models and the observed data. The key detail for computing the Validity Ratio is in selecing a suitable sampling strategy for causal models. Section~\ref{sec:nps-implementation} describes one such strategy based on the response variable framework \citep{balke1993nonparametric}.
For the rest of the paper, we assume that $X$, $Y$ and $Z$ are all discrete. In principle, we can apply the NPS testing framework to both continuous and discrete values for $X$, $Y$ and $Z$. However, the test is expected to be most informative when $X$ is discrete. This is because when $X$ is continuous, the region for Valid-IV identified by Pearl's necessary test coincides with the  Invalid-IV region and the necessary test becomes uninformative \citep{bonet2001}. 
For ease of exposition, we present the NPS test using binary $Z$, $X$ and $Y$. In Section~\ref{sec:extend-discrete}, we discuss how the test extends to the case where $Z$, $X$ and $Y$ can be arbitrary discrete variables. 
Finally, Sections~\ref{sec:simulations} and \ref{sec:applications} demonstrate the practical applicability of the test using simulation and data from past IV studies.




\section{A NECESSARY AND PROBABLY SUFFICIENT (NPS) TEST}
\label{sec:nps-test-theory}
The key idea behind the NPS test is that we can compare marginal likelihoods of valid-IV and invalid-IV class of models. To do so, we first present a Bayesian meta-model that describes how observed data can be generated from different values of the exclusion and as-if-random conditions. We then provide our main result that proposes a \textit{Validity Ratio} to compare valid-IV and invalid-IV models, followed by pseudo-code for an algorithm that uses the NPS test to validate an instrumental variable. 
Throughout, we assume that $Z$, $X$ and $Y$ are all discrete variables.

\subsection{Generating valid-IV and invalid-IV causal models }
\label{sec:generative-model}
As mentioned above, our strategy depends on simulating all causal models---both valid-IV and invalid-IV models---that could have generated the observed data. Therefore, we first describe a probabilistic generative \textit{meta-model} of how the observed data is generated from a causal model, which in turn, is generated based on the as-if-random and exclusion assumptions.

Let us first define the valid-IV and invalid-IV models formally in terms of the two IV assumptions: exclusion and as-if-random. A  valid IV model does not contain an edge from $Z \rightarrow Y$ or from $U \rightarrow Z$, as shown in Figure~\ref{fig:std-iv-model}a. This implies that both Exclusion and As-if-random conditions hold for a valid-IV model. 
Conversely, a causal model is an invalid IV model when at least one of Exclusion or As-if-random conditions is violated, as shown  by the dotted arrows in Figure~\ref{fig:std-iv-model}b. Therefore, given the causal structure $Z \rightarrow X \rightarrow Y$,  there are two classes of causal models that can generate observed data distributions over $X$, $Y$ and $Z$: 

\begin{itemize}
	\item Valid-IV model: $E=True$ and $R=True$
	\item Invalid-IV model: \textit{Not} ($E=True$ and $R=True$)
\end{itemize}
where $E$ denotes the exclusion assumption and $R$ denotes the as-if-random assumption.  

 Each of these classes of causal models---valid and invalid IV---in turn contains multiple causal models, based on the specific parameters ($\theta$) describing each edge of the graph.
This one-to-many relationship between conditions for IV validity and causal models can be made precise using a generative meta-model, as shown in Figure~\ref{fig:dgp-pgm}.  We show dotted arrows to distinguish this (probabilistic) generative meta-model from the causal models described earlier. 
The meta-model entails the following generative process: Based on the configuration of the Exclusion and As-if-Random conditions, one of the  causal model classes---Valid or Invalid IV---is selected. A specific model (\textit{Causal Model} node) is then generated by parameterizing  the selected class of causal models, where we use $\theta$ to denote model parameters. The  causal model results in a probability distribution over $Z$, $X$ and $Y$, from which observed data (\textit{Data} node) is sampled. Finally, we can apply Pearl-Bonet necessary test on the observed data, which leads to the binary variable $NecTestResult$.

For a given problem, we observe the data $D$ and result of the Pearl-Bonet necessary test. All other variables in the meta-model are unobserved. 

\begin{figure*}[t]
	\centering
	\includegraphics[scale=0.59]{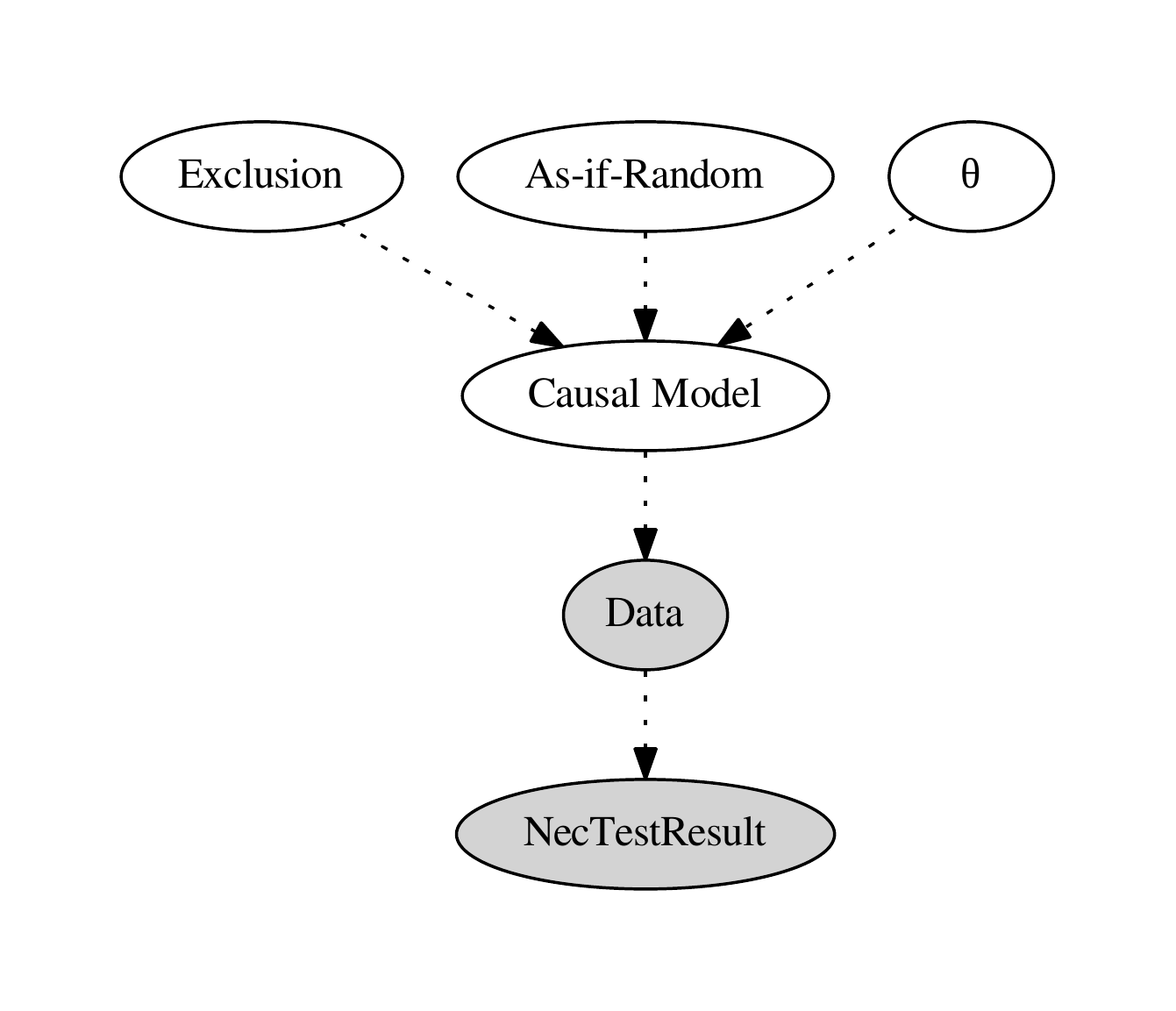}
	\caption{A probabilistic graphical meta-model for describing the connection between IV conditions and specific causal models. Evidence consists of both the results of the necessary test and observed data sample. Therefore, given this evidence, some causal models are expected to become more likely than others. Note that arrows are dotted to distinguish these \textit{probabilistic} diagrams from the causal diagrams in Figure 1.}
	\label{fig:dgp-pgm}
\end{figure*}

\subsection{Comparing marginal likelihood of Valid-IV and Invalid-IV model classes}
Let $PT$ denote whether the observed data passed the necessary test. We wish to estimate whether the data was generated from a Valid IV model. We can compare the likelihood of observing $PT$ and $D$ given that both Exclusion and As-if-random conditions are valid, versus when they are not. 

\begin{theorem}
	Given a representative data sample $D$ drawn from $P(X, Y, Z)$ over variables $X$, $Y$, $Z$, and result of the Pearl-Bonet necessary test $PT$ on the data sample, the validity of $Z$ as an instrument for estimating causal effect of $X$ on $Y$ can be decided using the following evidence-ratio of valid and invalid classes of models:
\begin{align} \label{eqn:model-ratio-mainthm}
	\textit{Validity-Ratio} &=\frac{P(E, R|PT, D)}{P(\neg (E, R)|PT,D)} =  \frac{P(PT,D|E, R)*P(E,R)}{P(PT,D|\neg (E, R))*P(\neg (E,R))} \nonumber \\
						     &= \frac{P(M1)}{P(M2)} \frac{\int_{M_1: m \text{ is valid}} P(m|E,R) P(D|m) dm}{\int_{M_2:m \text{ is invalid}} P(m | \neg (E,R)) P(D|m) dm}
\end{align}
where $M1$ and $M2$ denote classes of valid-IV and invalid-IV causal models respectively. $P(D|m)$ represents the likelihood of the data given a causal model $m$. $P(m|E,R)$ and $P(m | \neg (E,R)$ denote the prior probability of any model $m$ within the class of valid-IV and invalid-IV causal models respectively. 
\end{theorem}

While we are additionally using the result of the Pearl-Bonet necessary test to compute evidence, the Validity-Ratio reduces to the Bayes Factor \citep{kass1995bayes}. 
The proof of the theorem follows from the structure of the generative meta-model and properties of the Pearl-Bonet necessary test.

\begin{proof}
	Let us first consider the ratio of marginal likelihoods of the two model classes. 
\begin{align} \label{eqn:model-ratio}
	\textit{ML-Ratio} &= \frac{P(PT=1,D=d|E, R)}{P(PT=1,D=d|\neg (E, R))}
\end{align}

Since the Pearl-Bonet test is a necessary test, we know that $P(PT | E, R) = 1$ if the true data distributions are known. However, in practice, we will have a data sample and apply a statistical test. Therefore in some cases the test may return Fail even if $E$ and $R$ are satisfied, leading to the following expression for the numerator:
\begin{align} \label{eqn:model-ratio-num1}
P&(PT=1,D=d|E, R) \nonumber \\
 &=P(PT=1|D=d, E,R)P(D=d|E, R) \nonumber \\
 &=P(PT=1|D=d)P(D=d|E, R) \\
\end{align}

Further,  for any causal model $m$, we know with certainty whether it follows exclusion and as-if-random restrictions. In particular, $P(m_{invalidIV}|E, R) = 0$. Using this observation,  we can write $P(D|E,R)$ as:
\begin{align} \label{eqn:model-ratio-num2}
P&(D|E,R) \nonumber \\
&= \int_{m} P(D, m|E,R) dm \nonumber \\
&= \int_{m} P(m|E,R) P(D|m) dm \nonumber \\
&=  \int_{M_1:m \text{ is valid}} P(m|E,R) P(D|m) dm
\end{align}

Similarly, the denominator can be expressed by, 
\begin{align} \label{eqn:model-ratio-deno1}
P&(PT,D|\neg (E, R)) \nonumber \\
 &= P(PT|D, \neg (E, R)) P(D|\neg (E, R)) \nonumber \\
 &= P(PT|D) P(D|\neg (E, R)) \nonumber \\
 &= P(PT|D) \int_m  P( D, m| \neg (E, R)) dm  \nonumber \\
&= P(PT|D) \int_m  P(m| \neg (E,R)) P( D|m, \neg (E, R)) dm  \nonumber \\ 
 &=P(PT|D)  \int_{M_{2}:m \text{ is invalid}} P(m | \neg (E,R)) P( D|m) dm
\end{align}

where we use the conditional independencies entailed by the generative meta-model.

Combining Equations~\ref{eqn:model-ratio-num1}, \ref{eqn:model-ratio-num2} and \ref{eqn:model-ratio-deno1}, we obtain the ratio of marginal likelihoods: 
\begin{align} \label{eqn:model-ratio-reduced}
    \textit{ML-Ratio} &=	\frac{P(PT,D|E, R)}{P(PT,D|\neg (E, R))} = \frac{\int_{M1:m \text{ is valid}} P(m|E,R) P(D|m) dm}{\int_{M2:m \text{ is invalid}}  P(m| \neg (E,R)) P(D|m) dm}
\end{align}
Finally, by definition of model classes $M1$ and $M2$, they correspond to valid and invalid classes of causal models. Thus,
\begin{align}
	\frac{P(E,R)}{P(\neg (E, R))} &= \frac{P(M1)}{P(M2)}
\end{align}
The above two equations lead us to the main statement of the theorem:
\begin{align}
 \frac{P(PT,D|E, R)*P(E,R)}{P(PT,D|\neg (E, R))*P(\neg (E,R))}  &= \frac{P(M1)}{P(M2)}  \frac{\int_{M1:m \text{ is valid}} P(m|E,R) P(D|m) dm}{\int_{M2:m \text{ is invalid}} P(m| \neg (E,R)) P(D|m) dm}
\end{align}
\end{proof}

As with the Bayes Factor, estimation of the Validity-Ratio depends on the prior on causal models because the model is not identified. Since  the configuration of Exclusion and As-if-random conditions does not provide any more information apart from restricting the class of causal models, we can assume a uninformative uniform prior on causal models within each of the Valid-IV and Invalid-IV classes. If sufficient data is available, one may use the fractional Bayes Factor \citep{ohagan1995fractional} to split the sample and use the first subsample to find a prior on causal models using data likelihood, and the second to estimate the Validity Ratio. We discuss the effect of using other model priors in Section~\ref{sec:discussion}. Using a uniform model prior leads to the following corollary.

\begin{corollary} \label{cor:uniform-prior}
Using a uniform model prior $P(M1|E, R)$ for valid-IV models, $P(M2|\neg (E,R))$ for invalid-IV models, the Validity-Ratio from Theorem 1 reduces to:
\begin{align} \label{eqn:model-ratio-uniformprior}
\textit{Validity-Ratio} &=\frac{P(M1)}{P(M2)}  \frac{K_2 \int_{M1:m \text{ is valid}} P(D|m) dm}{ K_1 \int_{M2:m \text{ is invalid}} P(D|m) dm}
\end{align}
where $K_1$ and $K_2$ are normalization constants.
\end{corollary}

\subsection{NPS Algorithm for testing IVs}
\label{sec:nps-algorithm}
Based on the above theorem, we present the NPS algorithm for testing the validity of an instrumental variable below. Assume that the observational dataset contains values for three discrete variables: cause $X$, outcome $Y$ and a candidate instrument $Z$.

\begin{enumerate} 
\item Estimate $P(Y, X|Z)$ using observational data and run the Pearl-Bonet necessary test. If the necessary test fails, Return \textit{REJECT-IV}.
\item Else, compute the Validity-Ratio from Equation~\ref{eqn:model-ratio-mainthm} for the one or more of the following types of violations (can exclude violations that are known \emph{apriori} to be impossible):
	\begin{itemize}
		\item \textbf{Exclusion may be violated}: $Z \not \indep Y | X, U$
	\item \textbf{As-if-random may be violated}: $Z \not \indep U$
	\item \textbf{Both may be violated}:$Z \not \indep Y | X, U; Z \not \indep U$
	\end{itemize}
\item If all Validity Ratios are above a pre-determined threshold $\gamma$, then return \textit{ ACCEPT-IV}. Else if any Validity Ratio is less than $\gamma^{-1}$, then return \textit{REJECT-IV}. Else, return \textit{INCONCLUSIVE}.
\end{enumerate}


Although the NPS algorithm seems straightforward, in practice, the first two steps involve a number of smaller steps that we discuss in the next two sections. Section~\ref{sec:nps-implementation} describes how to compute the Validity Ratio for the three kinds of violations listed above and Section~\ref{sec:extend-discrete} discusses extensions of the Pearl-Bonet test that enable its empirical application to discrete variables.

\section{COMPUTING THE VALIDITY RATIO}
\label{sec:nps-implementation}
The key detail in implementing the NPS test is in evaluating the integrals in Equation~\ref{eqn:model-ratio-mainthm}, since there can be infinitely many valid-IV or invalid-IV causal models. 
In this section we first present the \emph{response variables} framework from  \cite{balke1994} that provides a finite representation for any non-parametric causal model with discrete $X$, $Y$ and $Z$. We then extend this framework to also work with invalid-IV causal models.
Armed with this characterization, we describe methods for computing the Validity-Ratio in Section~\ref{sec:compute-validity-ratio}.
Note that neither our finite characterization of causal models nor our methods for computing the integrals are unique; any other suitable strategy can be used to implement the NPS test.

\subsection{The response variable framework}
\label{sec:sample-causal-models}
One way  to characterize causal models is to assume specific functional relationships between observed variables, such as in the popular linear model for instrumental variables~\citep{imbens2015causal}. In most cases, however, the nature of the functional form is not known and thus parameterization in this way arbitrarily restricts the complexity of a causal model.  A more general way to make no assumptions on the functional relationships or the unobserved confounders, but rather reason about the space of all possible functions between observed variables. We will use this approach for characterizing valid-IV and invalid-IV causal models. 

As an example, suppose we observe the following functional relationship between Y and X, $y=k(x)$, where the true causal relationship is $y=f(x,u)$. 
Conceptually, the variables $U$ are simply additional inputs to this function.  
Thus the effect of unobserved confounders can be seen as simply transformation the observed relationship $k$ to another function $k'(x) = f(x, u)$. However, it is hard to reason about this transformation because the confounders are unobserved and may even be unknown. 
Here we make use of a property of discrete variables that stipulates a finite number of functions between any two discrete variables. Because the number of possible functions is finite, the  combined effect of unobserved confounders $U$ can be characterized by a finite set of parameters, \emph{without} making any assumptions about $U$. We will call these parameters \emph{response variables}, in line with past work. 
Note that we make no restriction on U---they can be discrete or continuous---but instead restrict the observed variables to be discrete.

More formally, a response variable acts as a selector on a non-deterministic function and converts it into a a set of deterministic functions, indexed by the response variable. Depending on the value of the response variable, one of the deterministic functions is invoked. Under this transformation, the response variables become the only random variables in the system, and therefore any  causal model can be expressed as a probability distribution over the response variables.

\begin{figure}[tb]
	\centering
	\subfloat[Valid IV]{\includegraphics[scale=0.6]{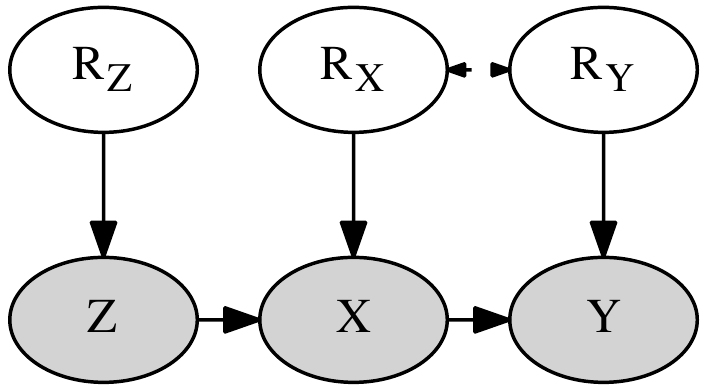}}%
	\qquad
	\subfloat[Invalid IV]{\includegraphics[scale=0.6]{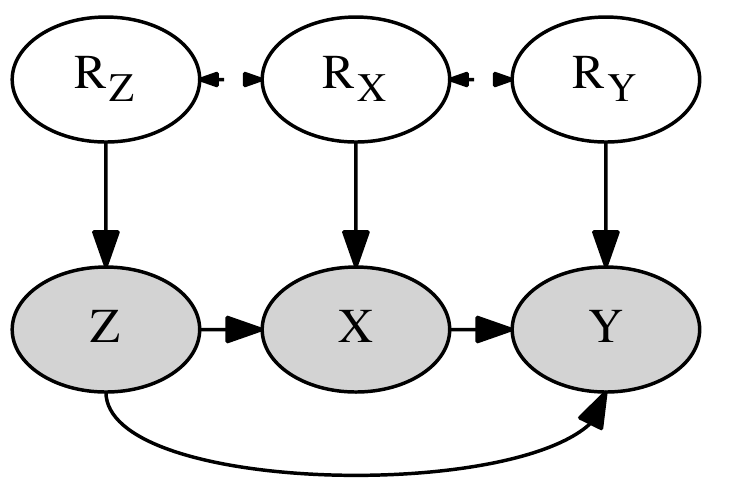}}%
	\caption{Causal graphical model with response variables denoting the effect of unknown, unobserved U.}
	\label{fig:iv-with-respvars}
\end{figure}

\subsubsection{Response variables framework for valid-IV models}
Let us first construct response variables for a valid-IV model, as in \cite{balke1994}. To do so, we will transform $U$ to a different \emph{response variable} for each observed variable in the causal model. For ease of exposition, we will assume that $Z$, $X$ and $Y$ are binary variables; however, the analysis follows through for any number of  discrete levels.

For valid-IV causal models, we can write the following structural equations for observed variables $X$, $Y$ and $Z$ (from Equation~\ref{eqn:structural-iv}).
\begin{align} \label{eqn:validiv-structural-eqns}
y &= f(x,u_y) \nonumber \\
x &= g(z,u_x) \nonumber \\
z &= h(u_z)
\end{align}
where $U_x$, $U_y$ and $U_z$ represent error terms. $U_x$ and $U_y$  are correlated. As-if-random condition ($Z \indep U$) stipulates that $U_z \indep \{U_y, U_x \}$. Exclusion condition is satisfied because function $f$ does not depend on $z$.

Since there are a finite number of functions between discrete variables, we can represent the effect of unknown confounders $U$ as a selection over those functions, indexed by a variable known as a response variable. For example, in Equation~\ref{eqn:validiv-structural-eqns}, $Y$ can be written as a combination of 4 deterministic functions of $x$, after introducing a response variable, $r_y$.
\begin{align} \label{eqn:validiv-structural-y}
y = 
\begin{cases} 
    f_{ry_0}(x) \equiv 0,& \text{if } r_y=0 \\
    f_{ry_1}(x) \equiv x,& \text{if } r_y=1\\
    f_{ry_2}(x) \equiv \tilde{x},& \text{if } r_y=2\\
    f_{ry_3}(x) \equiv 1,& \text{if } r_y=3
\end{cases} 
\end{align}

That is, different values of U change the value of Y from what it would have been without U's effect, which we capture through $r_y$.
Intuitively, these $r_y$ refer to different ways in which individuals may respond to the treatment $X$. Some may have no effect irrespective of treatment($r_y=0$), some may only have an effect when $X=1$ ($r_y=1$), some may only have an effect when X=0 ($r_y=2$), while others would always have an effect irrespective of $X$ ($r_y=3$). Such response behavior, as denoted by $r_y=\{0,1,2,3\}$, is analogous to \textit{never-recover}, \textit{helped}, \textit{hurt}, and \textit{always-recover} behavior in prior instrumental variable studies \citep{heckerman1995decision}.

Similarly, we can write a deterministic functional form for $X$, leading to the transformed causal diagram with response variables in Figure~\ref{fig:iv-with-respvars}a. 
\begin{align}  \label{eqn:validiv-structural-x}
x = 
\begin{cases}
	g_{rx_0}(z) \equiv 0,& \text{if } r_x=0 \\
	g_{rx_1}(z) \equiv z,& \text{if } r_x=1 \\
	g_{rx_2}(z) \equiv \tilde{z},& \text{if } r_x=2 \\
	g_{rx_3}(z) \equiv 1,& \text{if } r_x=3 
\end{cases}
\end{align}

Similar to $r_y$, $r_x=\{0,1,2,3\}$ can be interpreted in terms of a subject's compliance behavior to an instrument: \textit{never-taker}, \textit{complier}, \textit{defier}, and \textit{always-taker} \citep{angrist2008}.

Finally, $z$ can be assumed to be generated by its own response variable, $r_z$.
\begin{align}  
z = 
\begin{cases}
	0,& \text{if } r_z=0 \\
	1,&  \text{if } r_z=1
\end{cases}
\end{align}

Trivially, $Z=R_Z$.

Given this framework, a specific value of the joint probability distribution $P(r_z, r_x, r_y)$ defines a specific, valid causal model for an instrument $Z$.
Exclusion condition is satisfied because the structural equation for $Y$ does not depend on Z. For as-if-random condition, we additionally require that $U_z \indep \{U_x, U_y\}$. Since $R_X$ and $R_Y$ represent the effect of $U$  as shown in Figure~\ref{fig:iv-with-respvars}a, the as-if-random condition translates to $R_Z \indep \{ R_X, R_Y\}$, implying that $P(R_Z, R_X, R_Y)=P(R_Z)P(R_X, R_Y)$. Using this joint probability distribution over $r_z$, $r_x$, and $r_y$, any valid-IV causal model for $x$, $y$ and $z$ can be parameterized. For instance, when all three variables are binary, $R_Z$,  $R_X$ and $R_Y$ will be 2-level, 4-level and 4-level discrete variables respectively. Therefore, each unique causal model can be represented by 2+4x4=18 dimensional probability vector $\theta$ where each $\theta_i \in [0,1]$.
In general, for discrete-valued Z, X and Y with levels $l$, $m$ and $n$ respectively, $\theta$ will be a $(l+m^ln^m)$-dimensional vector.

\subsubsection{Response variable framework for invalid IVs}
\label{sec:resp-framework-invalid-iv}
While past work only considered Valid-IV models, we now show that the same framework can also be used to represent invalid-IV models. As defined in Section~\ref{sec:generative-model},  a causal model is invalid when either of the IV conditions is violated: Exclusion or As-if-random.

\subsubsection*{\textmd{Exclusion is violated}}
When exclusion is violated, it is no longer true that $Z \indep Y | X, U$. This implies that 
$Z$ may affect Y directly. To account for this, we redefine the structural equation for $Y$ to depend on both $Z$ and $X$: $y=h(X, Z)$. This corresponds to adding a direct arrow from $Z$ to $Y$ as shown in Figure~\ref{fig:iv-with-respvars}b. In response variables framework, this translates to:
\begin{align} \label{eqn:invalidiv-structural-y}
y = 
\begin{cases}
    h_{ry_0}(x, z) & \text{if } r_y=0 \\
    h_{ry_1}(x, z) & \text{if } r_y=1\\
    h_{ry_2}(x, z) & \text{if } r_y=2\\
    ... \\
    h_{ry_{15}}(x, z) & \text{if } r_y=15
\end{cases}
\end{align}

where $R_Y$ now has 16 discrete levels, each corresponding to a deterministic function from the tuple $(x, z)$ to $y$.


As with valid-IV causal models, any invalid-IV causal model can be denoted by a probability vector for $P(R_Z)$ and $P(R_X, R_Y)$. However, the dimensions of the probability vector will increase based on the extent of Exclusion violation. For full exclusion violation, dimensions will be $l+m^ln^{lm}$.

\subsubsection*{\textmd{As-if-random is violated}}
Violation of as-if-random does not change the structural equations, but it changes the dependence between $R_Z$ and $(R_X, R_Y)$. If as-if-random assumption does not hold, then $R_Z$ is no longer independent of $(R_X, R_Y)$. Therefore, we can no longer decompose $P(R_Z, R_X, R_Y)$ as the product of independent distributions $P(R_z)$ and $P(R_X, R_Y)$ and dimensions of $\theta$ will be $lm^ln^m$.

\subsubsection*{\textmd{Both exclusion and as-if-random are violated}}
In this case the structural equation for $Y$ is given by Equation~\ref{eqn:invalidiv-structural-y} and $R_Z$ is not independent of $(R_X, R_Y)$. Thus the dimensions of $\theta$ increase to $lm^ln^{lm}$.

\subsection{Computing marginal likelihood for Valid-IV and Invalid-IV models}
\label{sec:compute-validity-ratio}
The response variable framework provides a convenient way to specify an individual causal model:  choosing a causal model is equivalent to  sampling a probability vector $\theta$ from the joint probability distribution $P(r_x$, $r_y$, $r_z)$. 
The dimensions of this probability vector will vary based on the extent of violations of the instrumental variable conditions, from $l+m^ln^m$ for the valid-IV model to $lm^ln^{lm}$ for invalid-IV model in which both exclusion and as-if-random conditions are violated.
Below we describe how to compute the validity ratio for a given observed dataset.

\begin{figure}[h]
\begin{algorithm}[H]
	\KwData{Observed tuples (Z, X, Y), Prior-Ratio=$P(M1)/P(M2)$}
	\KwResult{Validity Ratio for comparing invalid and valid}
	Select appropriate subclass of invalid-IV models based on domain knowledge about the validity of IV conditions. \; 
	\begin{itemize}
		\item \textbf{Only Exclusion may be violated}: Assume $y=h(x, z, u)$. Sample $P(r_z)$, $P(r_x, r_y)$ separately. Use Equation~\ref{eqn:ml-exclusion-violated} to compute marginal likelihood $M_{EXCL}$.
		\item \textbf{Only As-if-random may be violated}: Assume $y=f(x, u)$. Sample $P(r_z, r_x, r_y)$ as a joint distribution. Use Equation~\ref{eqn:ml-air-violated} to compute marginal likelihood $M_{AIR}$.
		\item \textbf{Both conditions may be violated}: Assume $y=h(x, z, u)$. Sample $P(r_z, r_x, r_y)$ as a joint distribution. Use Equation~\ref{eqn:ml-both-violated} to compute marginal likelihood $M_{AIR,EXCL}$
	\end{itemize}
	
	Compute marginal likelihood of invalid-IV models as $ML_{INVALID} = max(M_{EXCL}, M_{AIR}, M_{AIR,EXCL})$ \;
	Compute marginal likelihood of valid-IV models using Equation~\ref{eqn:implement-numerator-valid-iv}, assuming $y=f(x, u)$ and sample $P(r_z)$, $P(r_x, r_y)$ separately \;
	Compute Validity Ratio as $ML_{VALID}/ML_{INVALID} * \textit{PRIOR-RATIO}$
	
	\caption{NPS Algorithm}
	\label{alg:nps}
\end{algorithm}
\caption{NPS Algorithm for validating an instrumental variable.}
\label{fig:alg-nps}
\end{figure}

To compute the Validity-Ratio, we return to Equation~\ref{cor:uniform-prior}.
\begin{align} \label{eqn:implement-ratio-valid-iv}
\textit{Validity-Ratio} &=\frac{P(M1)}{P(M2)}  \frac{K_2 \int_{M1:m \text{ is valid}} P(D|m) dm}{ K_1 \int_{M2:m \text{ is invalid}}  P(D|m) dm}
\end{align}

To compute the integrals in the numerator and the denominator of the above expression, we utilize the fact that there can be a finite number of unique observed data points (Z, X, Y) when all three variables are discrete. For example, for binary Z, X and Y, there can be 2x2x2=8 unique observations. In general, the number of unique data points is $lmn$.
Making the standard assumption of independent data points, we obtain the following likelihood for any causal model $m$,
\begin{align}
	P(D|m) &= \prod_{i=1}^{N} P(Z=z^{(i)}, X=x^{(i)}, Y=y^{(i)} | m )  \nonumber \\
	       &= (P(Z=0, X=0, Y=0 | m ))^{R_0}... P(Z=z_l, X=x_m, Y=y_n | m ))^{R_{lmn}} \nonumber \\
	       &=\prod_{j=1}^{Q} (P(Z=z_j, X=x_j, Y=y_j | m ))^{Q_j}
\end{align}
where N is the number of observed $(z, x, y)$ data points and $Q_j$ the number of times each unique value of $(z, x, y)$ repeats in the dataset. Since the model $m$ can be equivalently represented by its probability vector $\theta_{R_Z, R_X, R_Y}$, we can rewrite the above equation as:
\begin{align} \label{eqn:implement-general-ml}
	P(D|m) = P(D|\theta) &=\prod_{j=1}^{Q} (P(Z=z_j, X=x_j, Y=y_j | \theta ))^{Q_j} \nonumber \\
							 &= \prod_{j=1}^{Q} (\sum_{r_{zxy}=000}^{lmn} P(R_{XYZ} = r_{xyz}) P(Z=z_j, X=x_j, Y=y_j | \theta, r_{zxy} ))^{Q_j} 
\end{align}

For illustration, we derive the closed form expressions for the numerator and denominator of Equation~\ref{eqn:implement-ratio-valid-iv} when all variables are binary, in Appendix A. 
In general the method works for any number of discrete levels.


Finally, based on the above details,  Algorithm~\ref{alg:nps} summarizes the algorithm for computing validity ratio for any observed dataset.

\section{EXTENSIONS TO THE PEARL-BONET TEST}
\label{sec:extend-discrete}
In this section we describe extensions to the Pearl-Bonet test that are required for empirical application of the test for discrete variables. First, we present an efficient way to evaluate the necessary test for any number of discrete levels. Second, we show how to extend the monotonicity condition to more than two levels. Third, we discuss how to use the test in finite samples, by utilizing an exact test proposed by \cite{wang2016binaryivtest}.

\subsection{Implementing Pearl-Bonet test for discrete variables}
Specifying a closed form for the necessary test becomes complicated when we generalize from binary to discrete variables. \cite{bonet2001} showed that Pearl's instrumental inequalities for binary variables do not satisfy the existence requirement from Section~\ref{sec:related-work} and more inequalities are needed.  Further, it is not always feasible to construct analytically all the necessary inequalities for discrete variables. 

To  derive a practical test for IVs with discrete variables, we employ Bonet's framework  that specifies Valid-IV and Invalid-IV class of causal models as convex polytopes in multi-dimensional probability space.  
In Figure~\ref{fig:nectest-iv-schematic}, we showed a schematic of Bonet's framework, representing Valid-IV and Invalid-IV classes as polygons on a 2D surface. We now make these notions precise. The axes represent different dimensions of the probability vector $f=P(X, Y|Z)$. Assuming $l$ discrete levels for $Z$, $n$ for $X$ and $m$ for $Y$, $f$ will be a $lnm$ dimension vector, given by:
\begin{align}
f =& (P(X=x_1, Y=y_1| Z=z_1), \nonumber \\
& P(X=x_1, Y=y_2|Z=z_1),....  \nonumber\\
& P(X=x_1, Y=y_m| Z=z_1),  \nonumber\\
&P(X=x_1, Y=y_1|Z=z_2),...  \nonumber\\
&P(X=x_n, Y=y_m | Z=z_l)
)
\end{align}

$U$ may be either discrete or continuous, we do not impose any restrictions on unobserved variables. Any observed probability distribution over $Z$, $X$ and $Y$ can be expressed as a point in this $lmn$-dimension space. Since $\sum_{i,j} P(X=x_i, Y=y_j|Z=z_k) = 1 \forall k \in \{1...l\}$, the extreme points of for valid probability distributions $P(X, Y|Z)$ are given by $P(X=x_i, Y=y_j|Z=z_k)=1$. We showed a square region as the set of all valid probability distributions in Figure~\ref{fig:nectest-iv-schematic}, but more generally the region constitutes a $lmn$-dimensional convex polytope $\mathcal{F}$ \citep{bonet2001}.

Based on the models defined in Figure~\ref{fig:std-iv-model},  the set of all valid probability distributions $\mathcal{F}$ can be generated by Invalid-IV class of models. Within that set, we are interested in the probability distributions that can be generated by a valid-IV model. Knowing this subset provides a necessary test for instrumental variables; any observed data distribution that cannot be generated from a valid-IV model fails the test. Bonet showed that such a subset forms another convex polytope $\mathcal{B}$ (the Valid-IV region in Figure~\ref{fig:nectest-iv-schematic}) whose extreme points can be enumerated analytically. However, he did not provide an efficient way to check whether a given data distribution lies within this polytope. Below we provide an approach for checking this, thus yielding a practical implementation for a necessary test for discrete variables.

\begin{theorem} \label{thm:discrete-iv-test}
Given data on discrete variables $Z$, $X$ and $Y$, their observational probability vector $ f \in \mathcal{F}$, and extreme points of the polytope $B$ containing distributions generatable from a Valid-IV model, the following linear program  serves as a necessary and existence test for instrumental variables:
\begin{align} \label{eqn:discrete-iv-thm}
\sum_{k=1}^{K} \lambda_k.\vec{e_k} = \vec{f}; & \qquad \sum_{j=1}^{K} \lambda_j = 1 ; & \lambda_j \geq 0 \text{  }\forall j \in [1, K]
\end{align}
where $e_1, e_2...e_K$ represents $lmn$-dimensional extreme points of $B$  and $\lambda_1, \lambda_2..\lambda_K$ are non-negative real numbers. If the linear program has no solution, then the data distribution cannot be generated from a Valid-IV causal model.
\end{theorem}

\begin{proof}
The proof is based on properties of a convex polytope, which is also a convex set. A point lies inside a convex polytope if it can be expressed as a linear combination of the polytope's extreme points. Therefore, an observed data distribution could not have been generated from a Valid IV model if there is no real-valued solution to Equation~\ref{eqn:discrete-iv-thm}.
\end{proof}

While the test works for any discrete variables, in practice the test becomes computationally prohibitive for variables with large number of discrete levels, because the number of extreme points $K$ grows exponentially with $l$, $m$ and $n$. 
If the number of discrete levels is large, we can an entropy-based approximation instead, as in \cite{chaves2014inferring}.

\subsection{Extending Pearl-Bonet test to include Monotonicity}
Monotonicity is a common assumption made in instrumental variables studies, so it will be useful to extend the necessary test for discrete variables when monotonicity holds. No prior necessary test for monotonicity exists for discrete variables with more than two levels, so here we propose a test for monotonicity that can be used in conjunction with Theorem~\ref{thm:discrete-iv-test}.

As defined in Section~\ref{sec:rw-graphical-tests}, monotonicity implies that:
\begin{align} \label{eqn:discrete-monotonicity}
g(z_1, u) \geq g(z_2, u) \textit{ whenever } z_1 \geq z_2
\end{align}

That is, increasing $Z$ can cause $X$ to either increase or stay constant, but never decrease. Note that the above definition is without any loss of generality. In case $Z$ has a negative effect on $X$, we can do a simple transformation by inverting the discrete levels on $Z$ so that Equation~\ref{eqn:discrete-monotonicity} holds. 

By requiring this constraint on the structural equation between $X$ and $Z$, monotonicity restricts the observed data distribution. We use this observation to test for monotonicity.

\begin{theorem} \label{thm:discrete-monotonicity}
For any data distribution $P(X, Y, Z)$ generated from a  valid-IV model that also satisfies monotonicity, the following inequalities hold:
\begin{align} \label{eqn:discrete-monotonicity-thm}
P(Y=y, X \geq x| Z=z_0) \leq P(Y=y, X \geq x| Z=z_1) \quad ... \quad \leq P(Y=y, X \geq x| Z=z_{l-1}) \quad \forall x,y \nonumber\\
P(Y=y, X \leq x| Z=z_0) \geq P(Y=y, X \leq x| Z=z_1) \quad ... \quad \geq P(Y=y, X \leq x| Z=z_{l-1}) \quad \forall x,y
\end{align}
where $Z$, $X$ and $Y$ are ordered discrete variables of levels $l$, $n$ and $m$ respectively and $z_0 \leq z_1 ... \leq z_{l-1}$.
\end{theorem}

Proof of the theorem is in Appendix B. 
Note that for binary variables, Theorem~\ref{thm:discrete-monotonicity} reduces to $P(Y=y, X = 1| Z=z_0) \leq P(Y=y, X = 1| Z=z_1)$ and $P(Y=y, X =0| Z=z_0) \geq P(Y=y, X =0| Z=z_1)$, same as Equation~\ref{eqn:binary-monotonicity}.

\subsection{Finite sample testing for Pearl-Bonet test }
Finally, Pearl-Bonet test assumes that we can infer conditional probabilities $P(Y,X|Z)$ accurately. However, in any finite observed sample, we will only be able to compute a sample probability estimate. Therefore, we need a statistical test that accounts for the finite sample properties of any observed dataset. There are many tests proposed to deal with finite samples, such as by \cite{kitagawa2015-ivtest,huber2015ivtesting,wang2016binaryivtest,ramsahai2011}. In this paper we use an exact test proposed by  \cite{wang2016binaryivtest}, both for its simplicity and  because it makes no assumptions about the data-generating causal models. This test converts the inequalities of the necessary test into a version of one-tailed Fisher's exact test. As with all null hypothesis tests, the goal is to refute the null hypothesis. Here the null hypothesis is that the conditional probability distribution satisfies the inequalities of the Pearl-Bonet test. 
We then quantify the likelihood of seeing the obseved data under this null hypothesis, thus providing us with a p-value for the test. Because we are testing 4 inequalities at once, our analysis can be prone to multiple comparisons. Therefore, Wang et al. recommend a significance level of $\alpha/2$ for each test, where $\alpha$ is the desired significance level. 

However, this test does not work under monotonicity assumption. We extend their method for monotonicity, by using the same transformation to convert monotonicity inequalities to the Fisher's exact test. Again, to prevent false positives due to multiple comparisons, it would be ideal to choose a smaller significance level for each inequality, but the results we present are without any correction.

\section{SIMULATIONS: HOW POWERFUL IS THE NPS TEST?}
\label{sec:simulations}

We now report simulation results for the NPS test, with the goal of determining the extent to which NPS test can distinguish between a valid and invalid instrumental variable. We evaluate the two parts of the NPS test separately. First, we present simulations for testing the power of the Pearl-Bonet necessary test for different violations of the IV assumptions, and then move on to evaluating the Validity Ratio for simulated datasets.

\subsection{Evaluating the necessary test}
Since the Pearl-Bonet test is a necessary test, it can admit datasets generated by Invalid-IV models. Therefore, here our goal is to estimate how often data generated by Invalid-IV models pass the Pearl-Bonet test. To do so, we consider different Invalid-IV models and check whether observed data distributions generated from these models pass the test.  Instead of testing on a particular Invalid-IV model, we consider simulations over types of causal models that one might encounter in empirical IV studies.  That is, we consider different \emph{model configurations} based on key properties of an IV causal model such as the strength of the instrument, monotonicity of relationship between the instrument and outcome, and similarly between the cause and outcome. Under each of these model configurations, we would like to know how likely it is that an instrument is valid, given that it passes the Pearl-Bonet necessary test.

We consider model configurations in each of the three types of violations: exclusion, as-if-random, and both violated. For exclusion, we start with two cases---when the instrument has a monotonic effect on the outcome, and when both instrument and cause have a monotonic relationship with outcome---and construct different configurations as we vary the strength of these relationships. These model configurations are motivated by encouragement design studies~\citep{west2008encouragement,sovey2011ivreview}, where it is plausible to assume a non-decreasing relationship between the instrument and outcome. We then relax the non-decreasing assumption to study more general scenarios. For as-if-random, we consider configurations based on the strength of relationship between confounders and the instrument, or equivalently, between response variables for the instrument on the one hand,  and the cause and outcome on the other. Finally, we consider configurations where both exclusion and as-if-random are violated.  

For all model configurations, we sample uniformly at random $N=200$ invalid-IV causal models each. We do so by identifying the specific violation in each configuration and then using the appropriate conditions on response variables from Section~\ref{sec:resp-framework-invalid-iv}. In effect, we sample response variables to sample different causal models. Then, for each causal model, we generate the true probabilities, $P(X,Y|Z)$ required to evaluate the Pearl-Bonet test, using the corresponding response variables. We compute true probabilities instead of sampling a dataset, to eliminate errors in estimating probabilities from data (due to sampling variation in simulating a dataset of $<Z, X, Y>$ from each model). 
Finally, we compute the fraction of invalid-IV causal models that pass the Pearl-Bonet necessary test for each model configuration.

In addition to checking IV validity, we also estimate the resultant bias in the IV estimate, when using data distribution from an invalid-IV causal model.   It could be possible that an instrument is invalid in the strict sense defined above, but still provides a causal estimate with low bias.  
To estimate the causal effect $X \rightarrow Y$, we use the Wald estimator \citep{wald1940}, which for binary variables, can be written as \citep{balke1993nonparametric}:
$$\hat{W} = \frac{P(Y=1|Z=1) - P(Y=1|Z=0)}{P(X=1|Z=1)-P(X=1|Z=0)}$$
Because the causal effect between binary variables ranges from $[-1, 1]$, we bound the estimate within this interval. We will compare the effectiveness of the Pearl-Bonet test at different values of \emph{instrument strength}, which is defined as  $P(X=1|Z=1)-P(X=1|Z=0)$. Note that the instrument strength also appears in the denominator of the Wald estimator.

We characterize the response variables for Valid-IV and Invalid-IV models below.

\subsubsection*{\textmd{Valid-IV: Both conditions are satisfied}}

When both Exclusion and As-if-random conditions are satisfied, $R_Y$ is a 4-level discrete variable as in Equation~\ref{eqn:validiv-structural-y}. Because $R_Z \indep \{R_X, R_Y\}$, we sample $\theta_{R_Z}$ independently and separately sample the joint distribution vector of $\theta_{R_X, R_Y}$.

\subsubsection*{\textmd{Invalid-IV: At least one of the conditions is violated}}
When as-if-random assumption is not violated (such as when $Z$ is randomized), we sample $\theta_{R_Z}$ independently as for a Valid-IV model. However,  since Exclusion may be violated, $\theta_{R_X, R_Y}$ will be a 4x16-level discrete variable as in Equation~\ref{eqn:invalidiv-structural-y}. 
Otherwise, if Exclusion is not violated, then $\theta_{R_Y}$ remains a 4-level discrete variable. However, $R_Z$ is no longer independent of $(R_X, R_Y)$ because as-if-random may be violated. Therefore, we will sample a joint probability distribution vector $\theta_{R_Z, R_X, R_Y}$ from all possible joint probability vectors. Since our goal is to only generate invalid IV models, we reject any generated probability vector where $R_Z$ turns out to be independent of $R_X, R_Y$.  
When both conditions are violated, $R_Y$ is a 16-level discrete variable and $R_Z$ is not independent of $(R_X, R_Y)$. 
Thus, we sample a (2x4x16)-dimensional probability vector $\theta_{R_Z, R_X, R_Y}$. 

For all simulations, we assume that $Z$, $X$ and $Y$ are binary variables. Further, since monotonicity is a required assumption for obtaining  the interpretation of a local average causal effect \citep{angrist1994identification}, we also  assume monotonicity throughout. Below we present results separately for the three kinds of violations of the IV conditions: exclusion only, as-if-random only, and both violated.

\begin{figure}[t]
\centering
	\subfloat[Pearl-Bonet test with varying fraction of nondecreasing effect of Z on Y]{\includegraphics[width=0.45\textwidth]{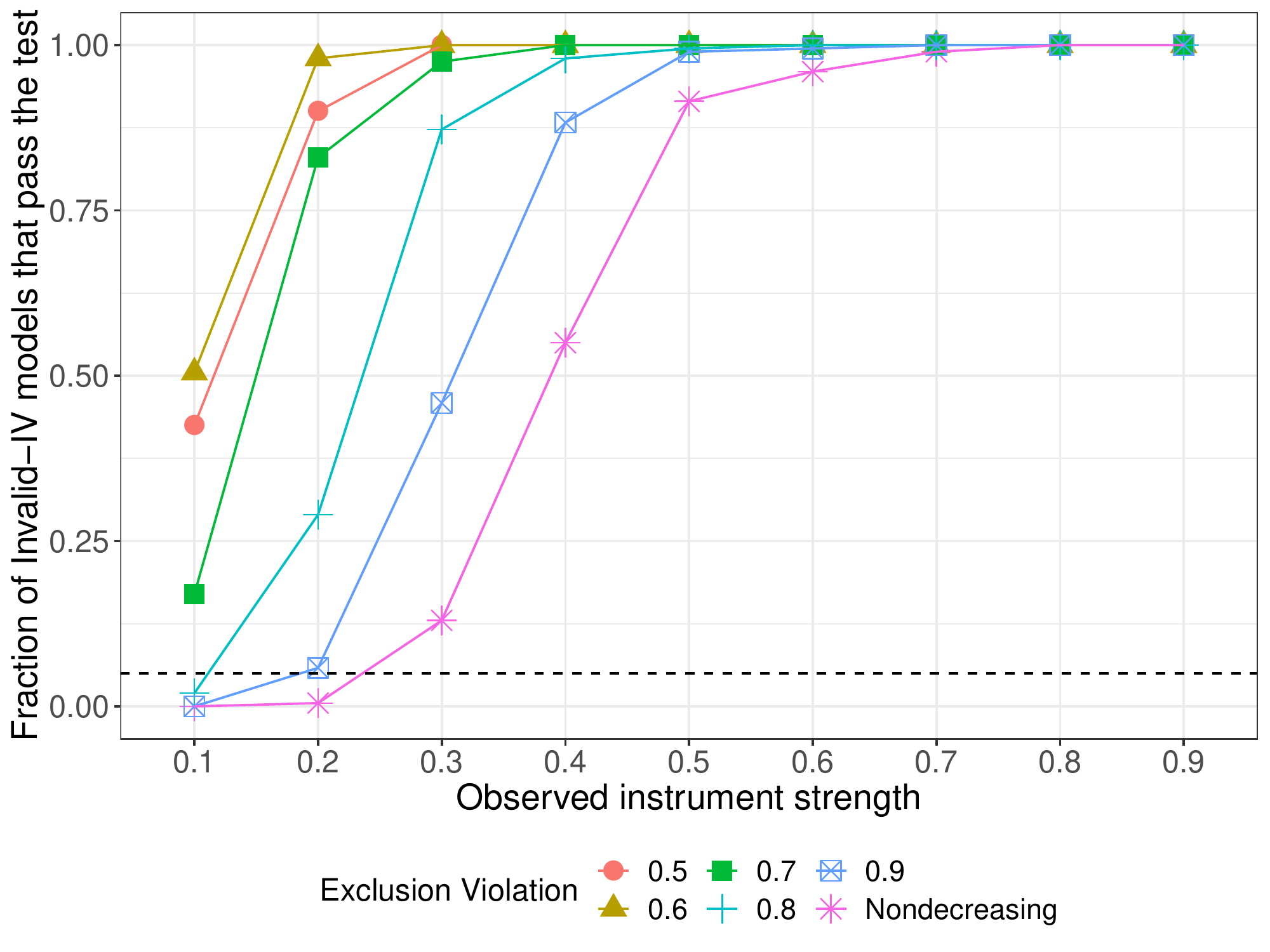}	\label{fig:fpr-nondecreasing-zonly}}%
	\qquad
	\subfloat[Corresponding bias of Wald estimator]{\includegraphics[width=0.45\textwidth]{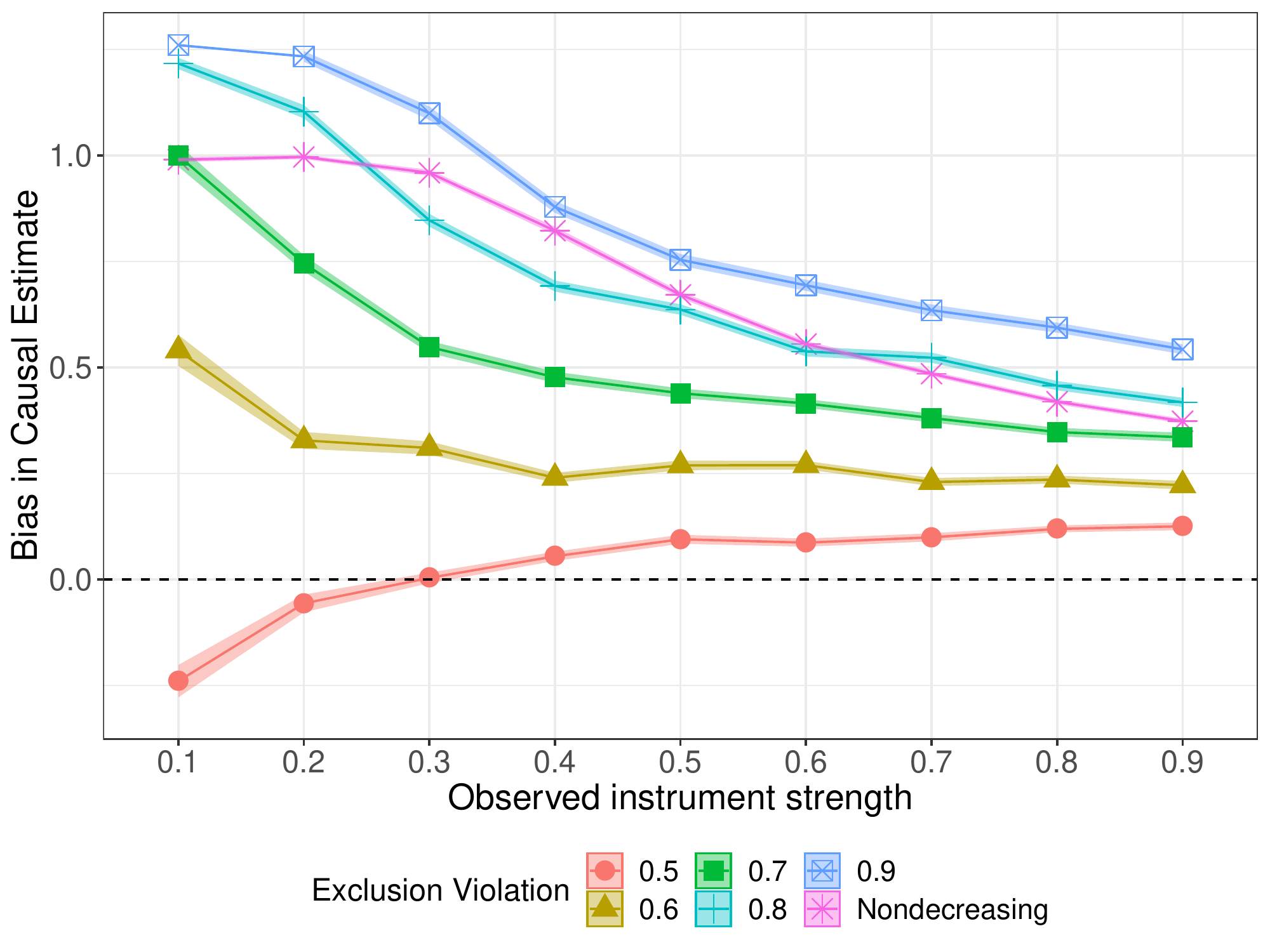} 	\label{fig:wald-nondecreasing-zonly}}

	\subfloat[Pearl-Bonet test with varying fraction of nondecreasing effect of both Z and X on Y]{\includegraphics[width=0.45\textwidth]{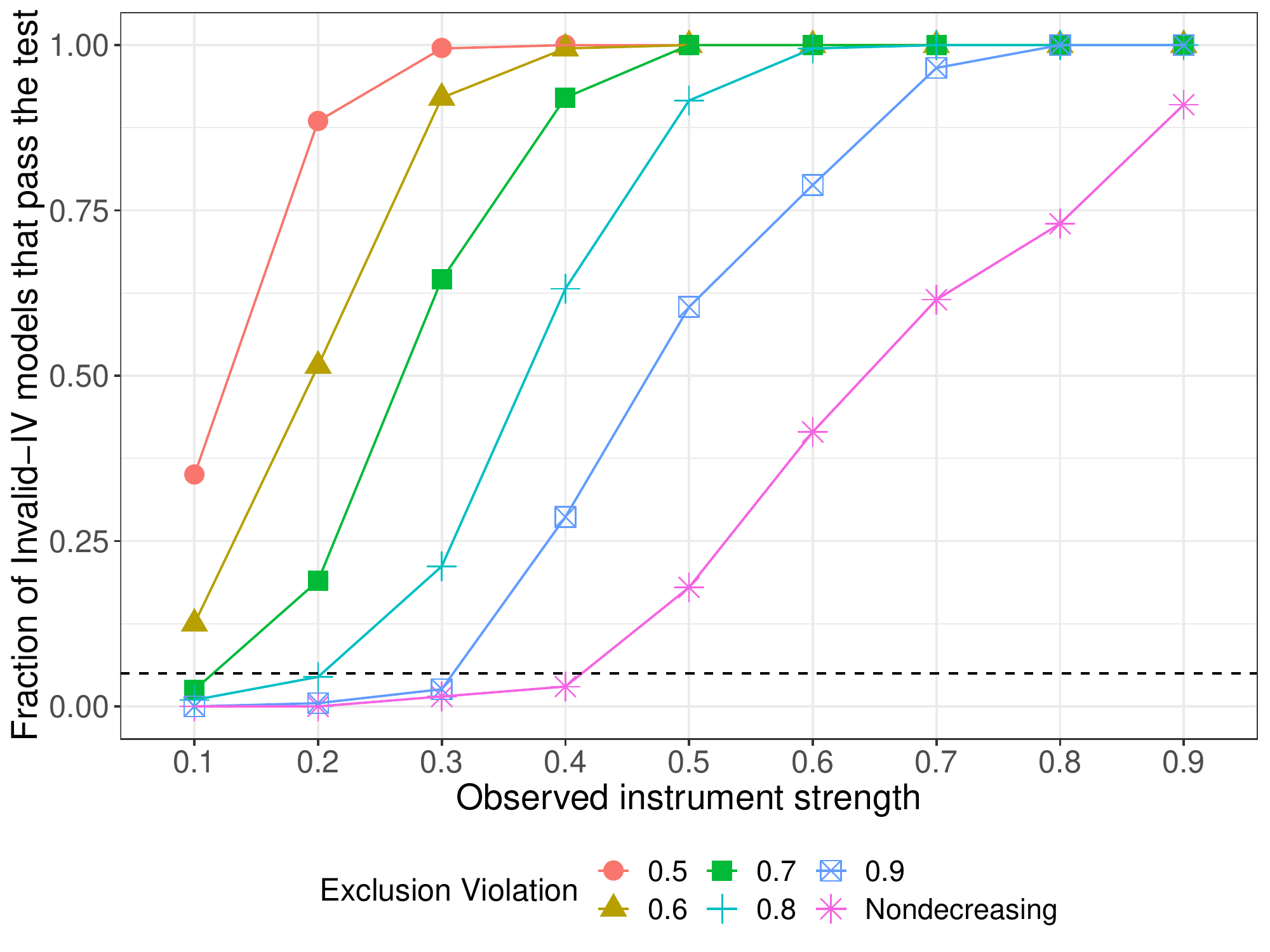}	\label{fig:fpr-nondecreasing}}%
	\qquad
	\subfloat[Corresponding bias of Wald estimator]{\includegraphics[width=0.45\textwidth]{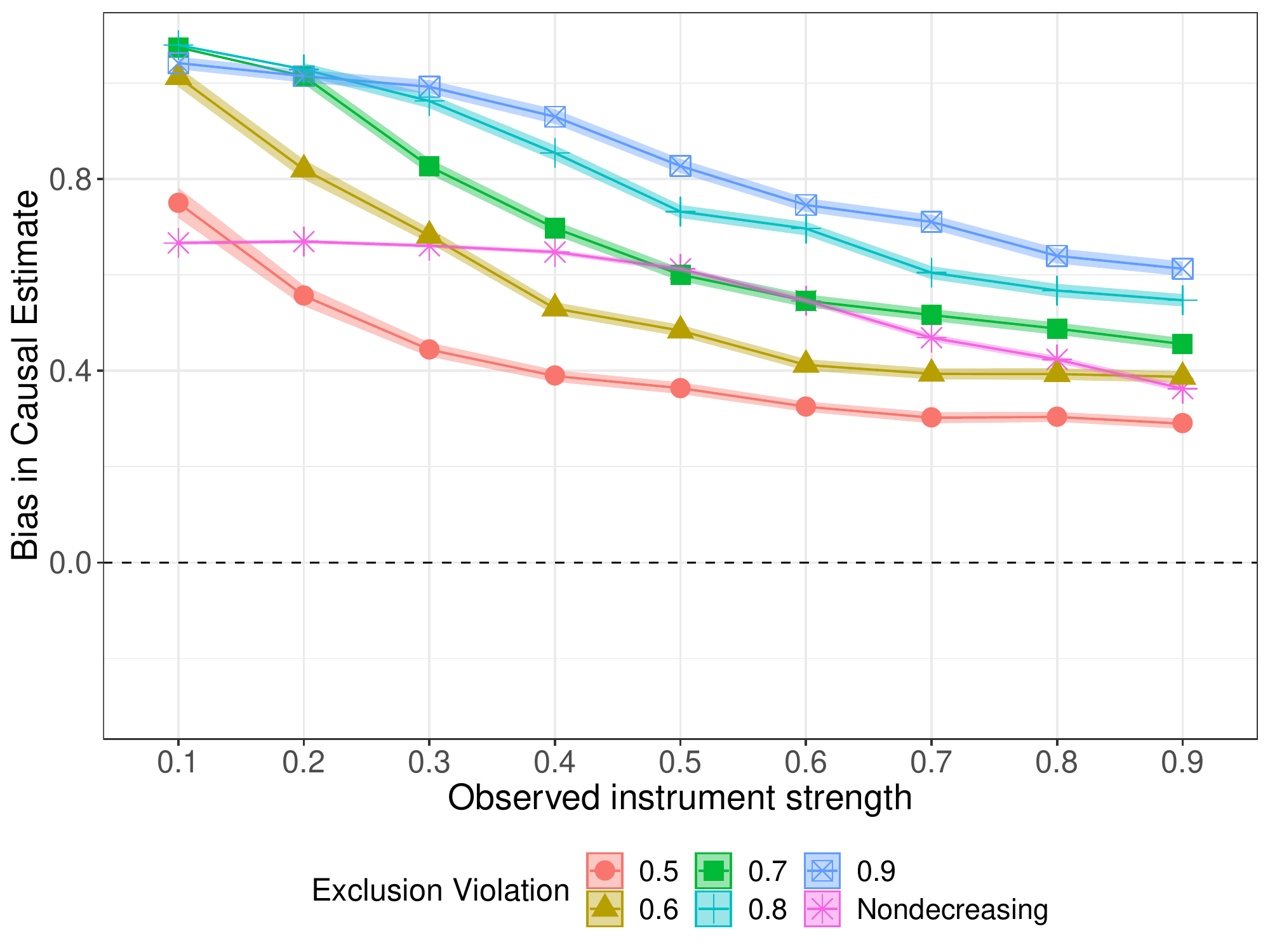} 	\label{fig:wald-nondecreasing}}
	\caption{Testing for the exclusion condition. Top panel corresponds to the scenario where $Z$ has a non-decreasing effect on $Y$, and the bottom panel corresponds to the scenario where both $Z$ and $X$ have a non-decreasing effect on $Y$. When the extent of nondecreasing-effect is 50\%, the test cannot filter out Invalid-IV models; however, the bias is also close to zero. The power of the test increases with increasing extent of the non-decreasing effect constraint. 
    In both scenarios, as the instrument strength is increased, the test becomes less effective. }

\end{figure}

\subsubsection{Exclusion may be violated}
We first test for violation of exclusion only. That is, we assume that as-if-random is satisfied (e.g. through random assignment). As described in Section~\ref{sec:resp-framework-invalid-iv}, violation of the exclusion restriction implies that $y=f(z, x, u)$ and thus there can be 16 different $r_y$ levels. Following the NPS Algorithm, we sample $r_x$ and $r_y$ jointly and sample $r_z$ independently. 

The remaining detail is how to sample invalid-IV models that violate exclusion condition. This is non-trival because the degree of violation of the exclusion restriction can vary based on known properties of the underlying causal model.  We take one of the 
simplest properties of the unobserved true causal model---the direction of effect from Z or X to Y---and characterize the power of the Pearl-Bonet test as we vary this property. First we will consider a scenario where either of $Z$ or $X$  has a non-decreasing effect on $Y$ and then gradually weaken this requirement to obtain a more general scenario.  This is motivated by encouragement design IV studies where we may know apriori that $Z$ or $X$ have a non-decreasing effect on $Y$ because there is no plausible mechanism that leads to a decrease in $Y$ with increase in $Z$ or $X$. 

However, in  other studies, completely ruling out model configurations where the non-decreasing property does not hold is too strong a condition.  We therefore relax this restriction and instead stipulate the percentage of data points where this restriction is satisfied. 
We do so by fixing the probability of the relevant response variables which correspond to a nondecreasing effect.
Driving this percentage down to 50\% essentially provides the general case, where the direction of the effect from Z to Y is equally likely to be positive or negative.

\subsubsection*{Extent of non-decreasing effect of instrument on outcome}
First, we simulate the model configurations where exclusion is violated by varying the extent to which $Z$ has a non-decreasing effect on $Y$. Let us first consider the scenario where the non-decreasing effect is strict: the instrument cannot cause the outcome to decrease. Figure~\ref{fig:fpr-nondecreasing-zonly} shows the fraction of such Invalid-IV models that pass the test at varying  observed instrument strength (marked as the `Nondecreasing' line). At low instrument strength, less than 5\% of invalid-IV models pass the necessary test. 
This simulation result indicates that if we observe empirically that a weak instrument passes the necessary test, it is likely to be a valid instrument assuming that $Z$ cannot decrease $Y$. However, as instrument strength increases, utility of the Pearl-Bonet test decreases. 

To relax the non-decreasing assumption, we simulate $\alpha$-non-decreasing cases such that $\alpha$ fraction of the units have a non-decreasing relationship; the rest do not. When $\alpha=0.5$, there is an equal chance of Z decreasing or increasing the value of Y. Each of the different lines in Figure~\ref{fig:fpr-nondecreasing-zonly} corresponds to a class of Invalid-IV models with a specified probability of non-decreasing effect between Z and Y. 
We observe that the utility of the Pearl-Bonet test decreases as $\alpha$ decreases. Even at  $\alpha=0.8$, however, the test is only able to identify invalid-IV models with less than $5\%$ error for instruments with strength less than or equal to $0.1$. When non-decreasing relationship does not exist (e.g., $\alpha=0.5$), the test performs poorly and cannot identify an invalid-IV model reliably. Thus, as we relax the strictness of the non-decreasing effect assumption, more invalid-IV models are passed by the necessary test and thus the test remains inconclusive. 

Contrasting these results with estimated bias in the Wald estimate of the causal effect $X\rightarrow Y$ provides more context. Even when the test is unable to detect violation of exclusion, it is also likely that the bias is relatively low (Figure~\ref{fig:wald-nondecreasing-zonly}). The magnitude of the bias is larger for weak instruments and for higher values of $\alpha$, both scenarios where the Pearl-Bonet test is the most discriminative. When the observed strength of the instrument is high, even clearly invalid-IV models lead to comparatively lower bias (less than $0.6$).

\subsubsection*{Extent of non-decreasing effect of both instrument and cause on outcome}
In some IV studies, we may know that both instrument $Z$ and cause $X$ have a non-decreasing effect on the outcome.  In such cases, we can strengthen the above assumption by assuming that both $Z$ and $X$ have a non-decreasing effect on $Y$.  
Figure~\ref{fig:fpr-nondecreasing} shows the fraction of invalid-IV models that pass the Pearl-Bonet test under these conditions. When the non-decreasing condition is satisfied strictly---that is, neither Z nor X can have an increasing effect on Y---less than 5\% of invalid-IV models pass the test up to a maximum instrument strength of $0.4$. 
The discriminatory power of the Pearl-Bonet test for other scenarios also increases compared to Figure~\ref{fig:fpr-nondecreasing-zonly}. For thresholds $\alpha$ of non-decreasing effect at least $0.7$, fraction of incorrectly identified Invalid-IV models lies below $5\%$ at instrument strength of 0.1. Similar to the previous results for bias, Figure~\ref{fig:wald-nondecreasing} shows that bias is highest for weak instruments or when the probability of having a non-decreasing effect is the highest. Happily, in both of these situations, the Pearl-Bonet test is the most effective at filtering out Invalid-IV models.

The lack of Pearl-Bonet test's effectiveness with a strong instrument $Z$ is not surprising: in the limit, $Z$ could be identical to X (an experiment with full compliance) and then Pearl-Bonet test inequalities (Equation~\ref{eqn:binary-monotonicity}) are satisfied trivially because the RHS will be 0. Clearly, these inequalities will be most discriminative when the instrument is weak. As we will see, this pattern will be consistent in the all results we obtain. Similarly, we saw that bias in the causal estimate is highest for weak instruments; this trend also repeats across our simulations, consistent with past results that show even small violations in IV conditions can lead to big finite sample biases in the IV estimates, especially when the instrument is weak \citep{bound1995-ivproblems}.

\begin{figure}[t]
\centering
	\subfloat[Pearl-Bonet test at varying severity of as-if-random violation]{\includegraphics[width=0.45\textwidth]{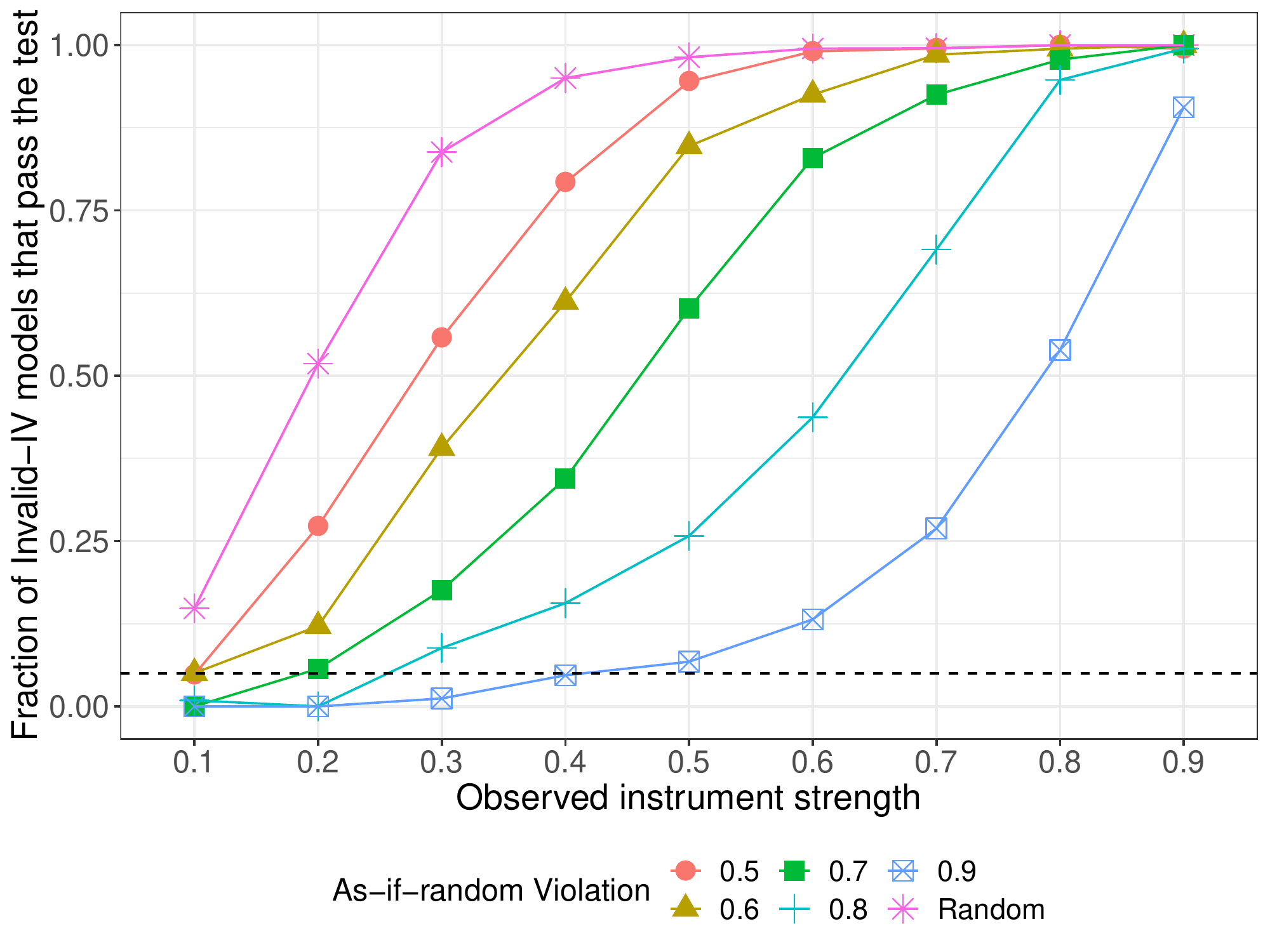}	\label{fig:fpr-asifrandom}}%
	\qquad
	\subfloat[Corresponding bias of Wald estimator]{\includegraphics[width=0.45\textwidth]{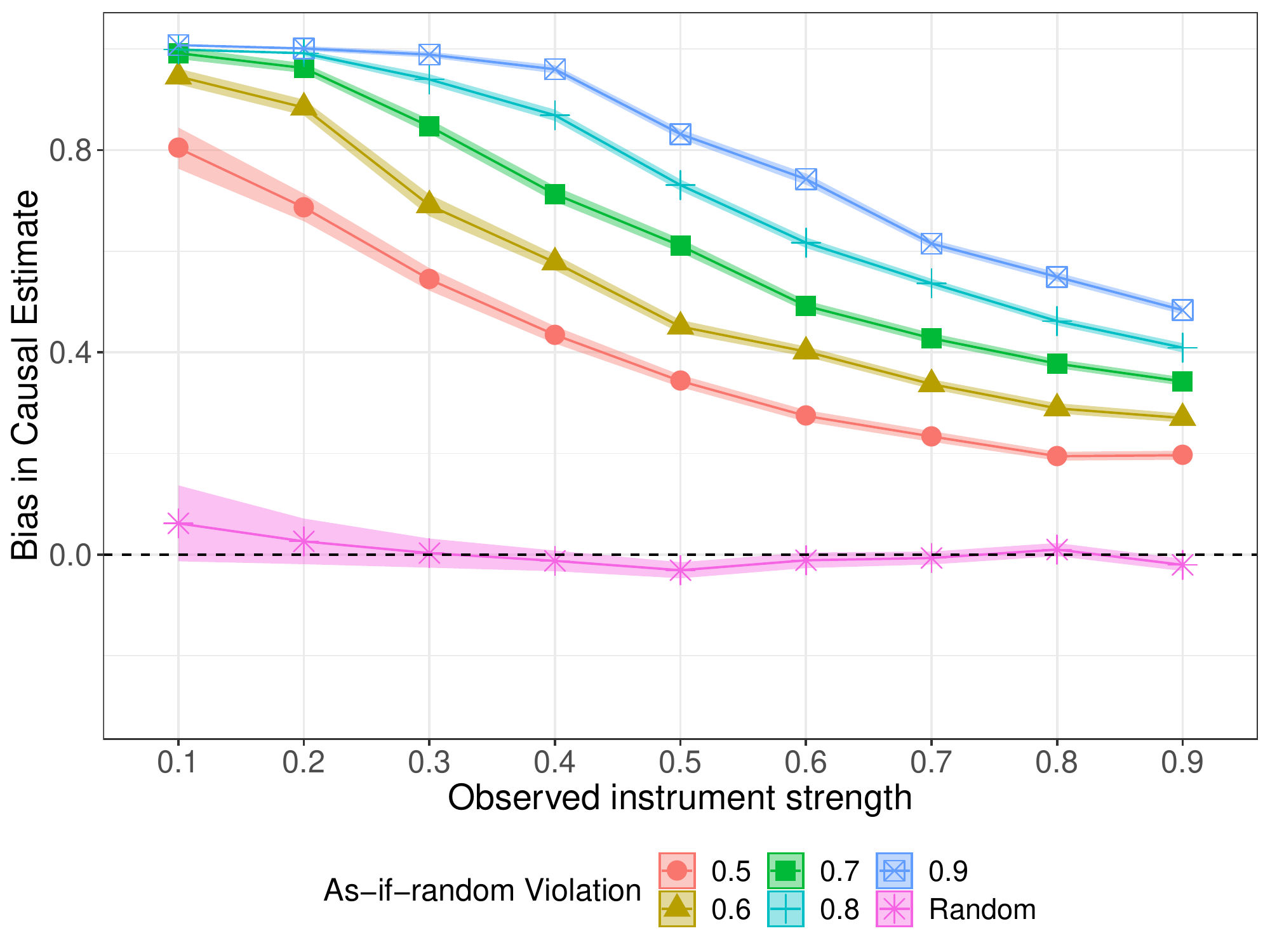} 	\label{fig:wald-asifrandom}}
	\caption{Testing for the as-if-random condition. Varying the mutual information between $R_Z\text{-}R_Y$. As the severity of as-if-random violation---mutual information---is increased, power of the Pearl-Bonet test in detecting Invalid-IV models increases and so does bias in the resultant IV estimate. 
	}

\end{figure}

\subsubsection{As-if-random may be violated}\label{sec:uniform-asifrandom-res}
Violation of the as-if-random condition implies that $r_z$ is not independent of $r_x$ and $r_y$. Here we assume that Exclusion is not violated. Following Algorithm~\ref{alg:nps}, we generate a joint distribution over $r_z$, $r_x$ and $r_y$ variables, sampling them uniformly at random. 
As with the exclusion restriction, there can be a number of ways to define the strength of an as-if-random violation, depending on how we specify the dependence between $R_Z$,  $R_X$ and $R_Y$. 
For the results presented, we define the strength of the violation as the mutual information between $r_x$ and $(r_x, r_y)$. When as-if-random is satisfied, mutual information will be zero.  
As we increase the mutual information, violation of as-if-random is expected to become more and more severe. 
Since correlation between $R_Z$ and $R_Y$ is necessary and sufficient for a violation of the as-if-random condition (but not correlation between $R_Z$ and $R_X$), we modulate the severity of violation by changing the correlation between $R_Z$ and $R_Y$. To do so, we vary a single conditional probability, $P(r_y=3|r_z=1)$ for simplicity; similar results can be obtained by varying other probabilities. We choose $P(r_y=3|r_z=1)$ because of the intuitive property that when it is high, Z and Y will also be highly correlated. 

Figure~\ref{fig:fpr-asifrandom} shows the results of the Pearl-Bonet test when as-if-random condition is violated. When we sample $P(R_Z, R_Y, R_Y)$ uniformly at random, the test is unable to distinguish effectively between invalid-IV and valid-IV models. Even at low  values of instrument strength ($<=0.2$), nearly half of invalid-IV models pass the Pearl-Bonet test. 
However, we also see the Wald estimator is reasonably accurate at all levels of instrument strength, even though the as-if-random condition is not satisfied. This indicates  that complete uniform sampling of causal models does not introduce a strong enough as-if-random violation to be either detected by the test or result in a noticeable biased estimate.

When the mutual information is increased beteen $R_Z$ and $R_Y$, we find that the discriminatory power of the Pearl-Bonet test increases. When the mutual information threshold is $>=0.7$, instruments with strength up to 0.2 have an errror rate of roughly $5\%$. 
Thus, the test is more powerful for stronger violations of the as-if-random assumption. 
That said, the test is unable to capture all violations that lead to noticeable bias. For instance, at a threshold of $0.5$, bias in the causal estimate can be as high as $0.5$, but the necessary test is unable to detect invalid-IV models more than $50\%$ of the time.

\begin{figure}[t]
\centering
\subfloat[Pearl-Bonet test with uniformly sampled models that violate both conditions]{\includegraphics[width=0.45\textwidth]{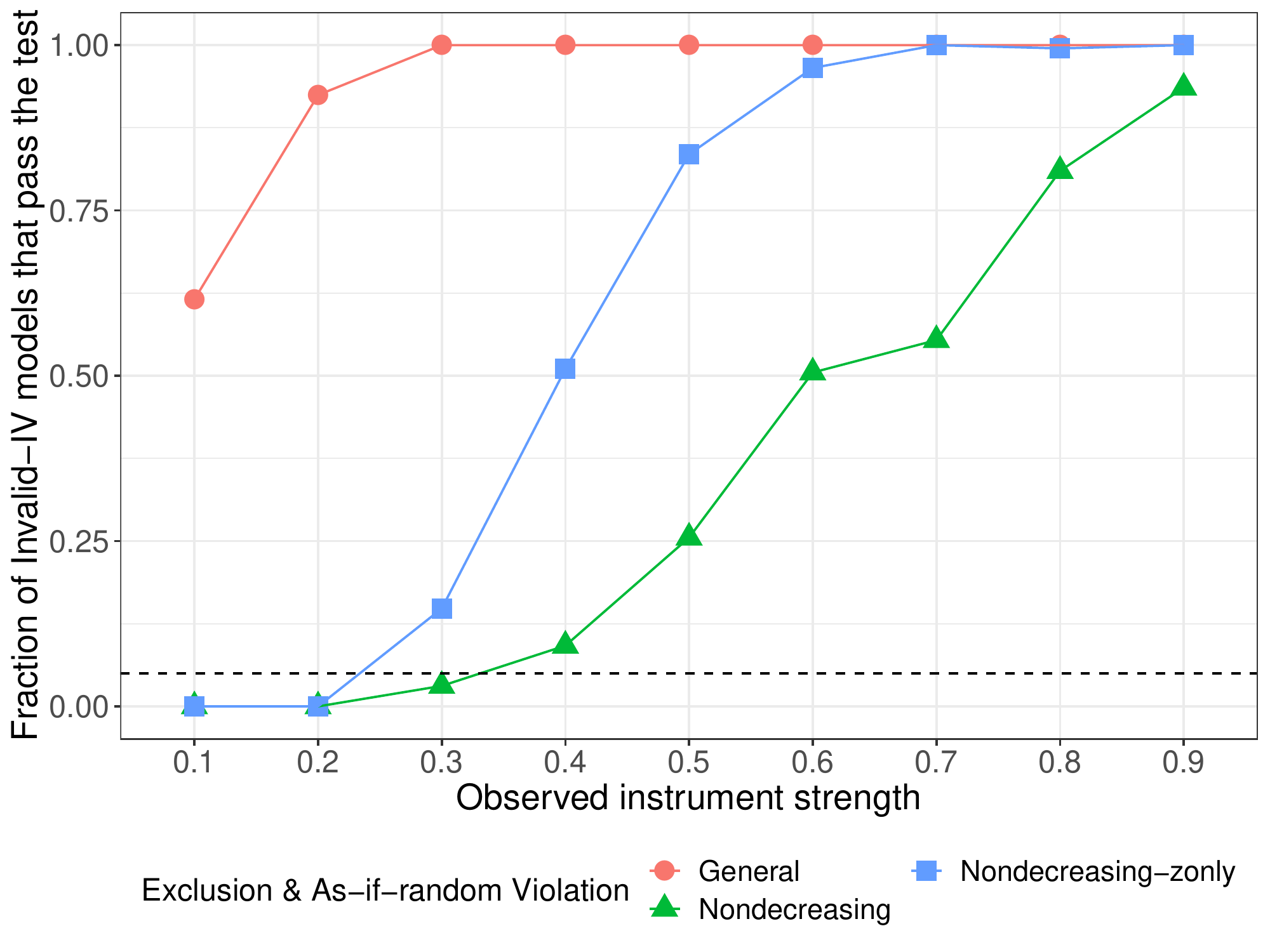}	\label{fig:fpr-both-exclusion}}%
	\qquad
	\subfloat[Corresponding bias of Wald estimator]{\includegraphics[width=0.45\textwidth]{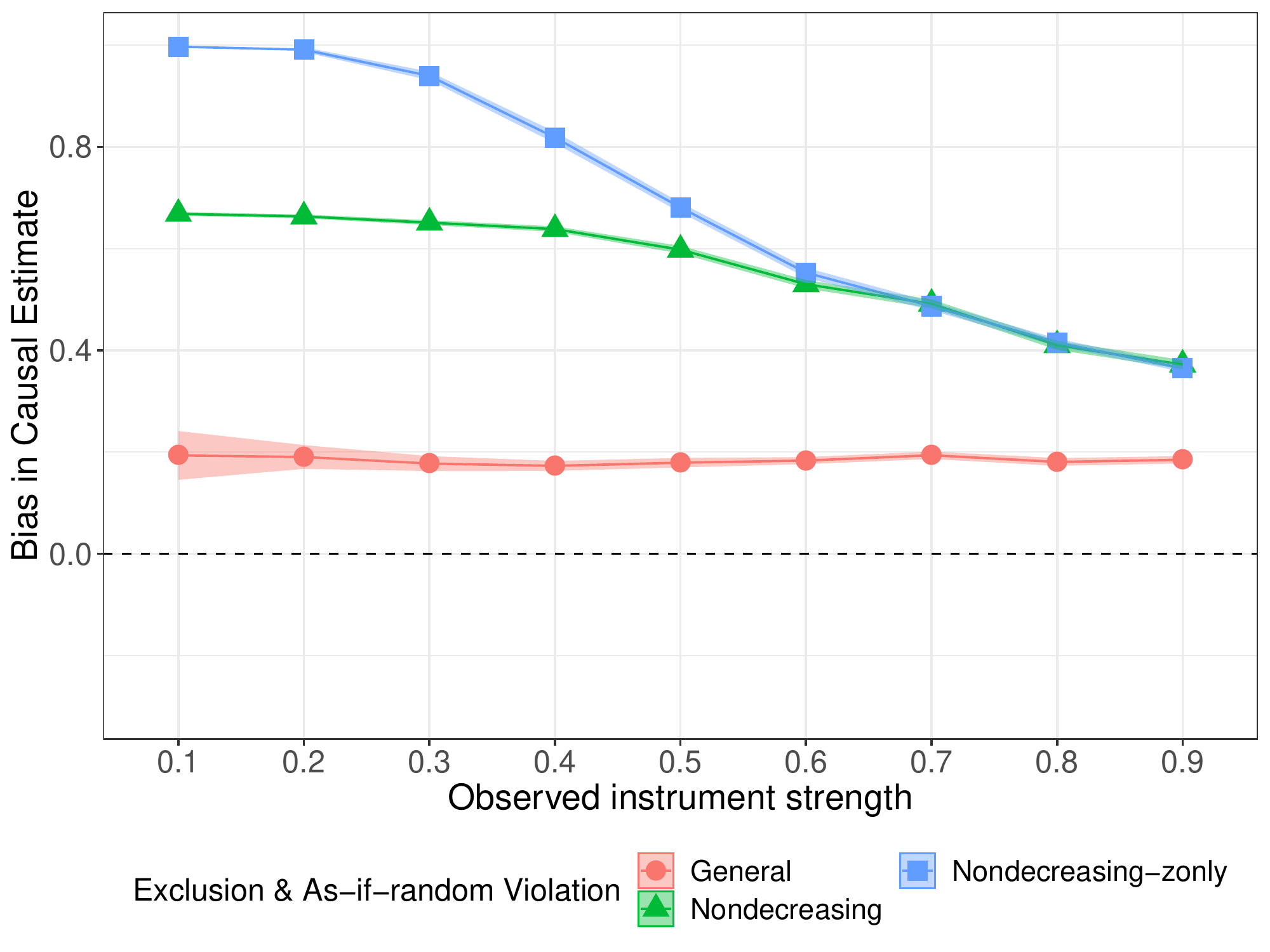} 	\label{fig:wald-both-exclusion}}

	\subfloat[Pearl-Bonet test with varying severity of as-if-random violation]{\includegraphics[width=0.45\textwidth]{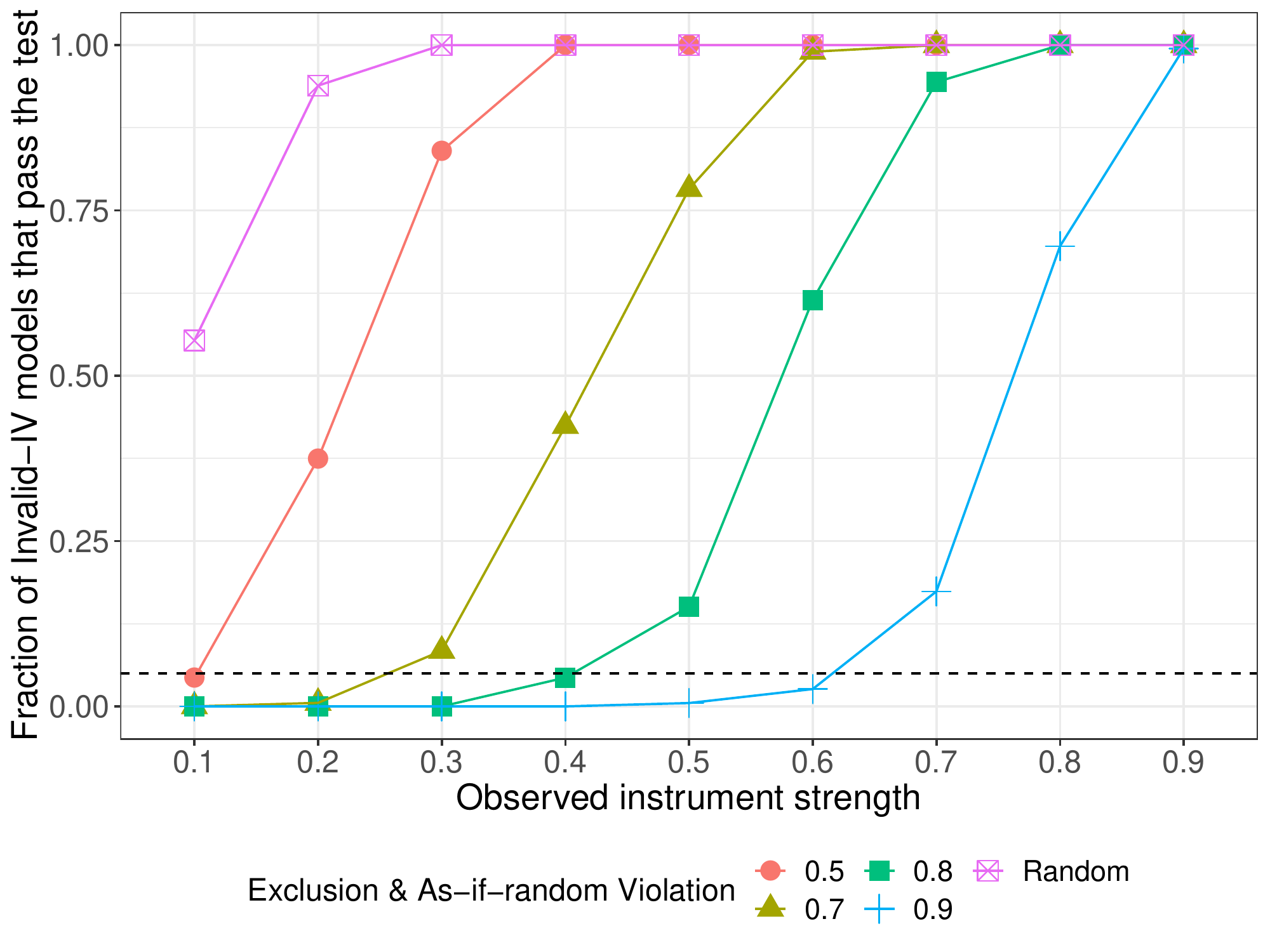}	\label{fig:fpr-both-asifrandom}}%
	\qquad
	\subfloat[Corresponding bias of Wald estimator]{\includegraphics[width=0.45\textwidth]{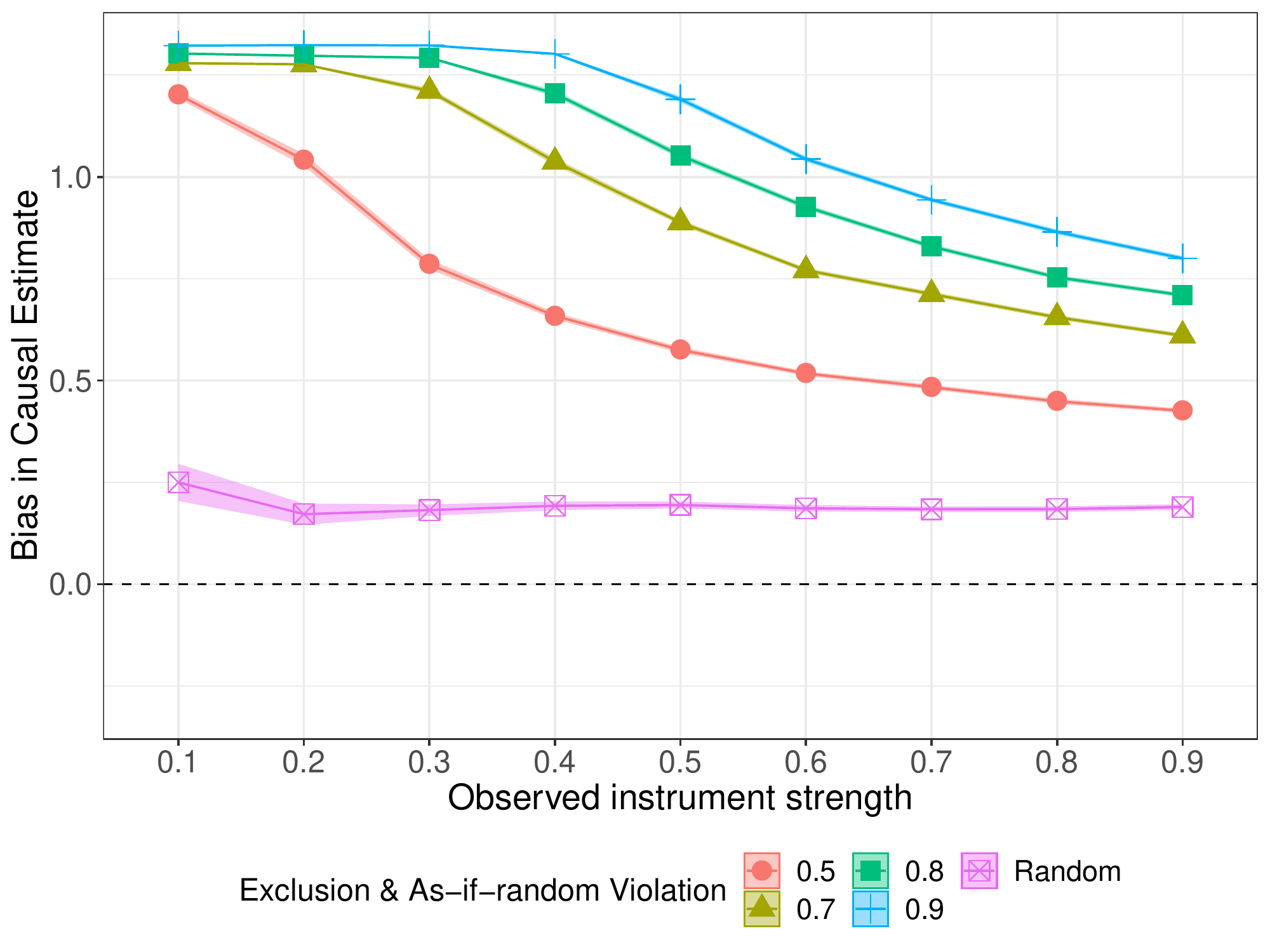} 	\label{fig:wald-both-asifrandom}}
    \caption{Both violated: Exclusion and as-if-random. Top panel corresponds to Invalid-IV models that violate both conditions. Different lines represent whether $Z \rightarrow Y$ is nondecreasing, $Z\rightarrow Y$ and $X\rightarrow Y$ are nondecreasing, or that none of them are nondecreasing. The test is more effective with  nondecreasing-effect constraints.   The bottom panel shows the power of the Pearl-Bonet test as the severity of as-if-random violation is increased. Stronger violations are more likely to be filtered out by the Pearl-Bonet test.}
\end{figure}

\subsubsection{Both exclusion and as-if-random may be violated}
Finally, we consider the case when both conditions may be violated. First, let us consider the straightforward case when both exclusion and as-if-random violating models are uniformly sampled. The red line in Figure~\ref{fig:fpr-both-exclusion} shows that the Pearl-Bonet  test performs poorly for such invalid-IV models; even for weak instrument strength, the test can correctly identify invalid-IV models less than $40\%$ of the time. Fortunately, the bias is also negligible (Figure~\ref{fig:wald-both-exclusion}), indicating that uniform violation does not lead to a noticeable bias in causal IV estimates.

When we stipulate that Z can only have a non-decreasing effect on Y, as in the above subsection, discriminatory power of the test improves. The error rate for identifying invalid-IV models is less than $5\%$ for instruments with strength up to 0.2. However, the bias also shoots up. When the observed instrument strength is high (say 0.5), bias in the causal estimate is over 0.6, but the necessary test can correctly identify an invalid-IV model less than $20\%$ of the time. Assuming that both $Z$ and $X$ have a non-decreasing effect on Y provides slightly better results. Instruments of strength up to $0.4$ have error rates nearly 5\% and the bias also decreases.

When we modulate the severity of the as-if-random condition (while keeping exclusion violation uniformly at random), the utility of the Pearl-Bonet test improves substantially (Figure~\ref{fig:fpr-both-asifrandom}). For thresholds at least $0.7$, the test misses less than $5\%$ of the invalid-IV models at an instrument strength of $0.2$. Bias is also high for these thresholds, but we have a higher chance of correctly filtering out invalid-IV models.

\subsubsection{Summary of results}
Two key patterns emerge. First, the test is more powerful in recognizing violations that also lead to a substantial bias. This is encouraging because the kind of violations that bias the causal estimate are exactly the invalid-IV datasets we want to eliminate. On average, these results suggest that when the Pearl-Bonet test is inconclusive, it is unlikely that applying the Wald Estimator will lead to substantial bias in the causal estimate. Conversely, in cases when the Pearl-Bonet test has high discriminatory power, eliminating invalid-IV data models will avoid computing Wald Estimates with high bias. 

Second, the above results show that detecting the violation of IV conditions is sensitive to the strength of the instrument. This may seem as a big limitation; however, in most observational studies, instruments with high strength are rare. For instance, in economics, \citep{staiger1994instrumental} recommend an F-value of $>10$ to prevent weak instrument bias. At such a low value for $F$, the instrument is likely to have to low correlation with X. Similarly, in epigenetics, $R^2$ of 0.1 between Z and X is typical \citep{pierce2010power}. 
In these low strength regimes, the Pearl-Bonet test can be effective in  testing for validity of an instrumental variable. 

\subsection{An example open problem for binary instrumental variables}
The above results are indicative and  they do not provide evidence on how the NPS test will perform on a particular single dataset.
Moreover, computing the Validity-Ratio for datasets that pass the Pearl-Bonet test can further improve effectiveness in detecting Valid-IV models.  Therefore, we now simulate datasets and check whether estimating the Validity Ratio can correctly identify whether they contain a valid instrumental variable or not.   To start with, we consider the following causal model from \cite{palmer2011ivbounds} where the Pearl-Bonet necessary test fails to detect violation of IV assumptions. 

\begin{align}
    \centering
    Z \sim & Bern(0.5) \nonumber \\
    U \sim & Bern(0.5) \nonumber\\
    X \sim & Bern(p_X); p_X = 0.05 + 0.1Z+0.1U \nonumber\\
    Y_0 \sim & Bern(p_0); p_0=0.1 + 0.05X + 0.1U \nonumber\\
    Y_1 \sim & Bern(p_1); p_1=0.1 + 0.2Z + 0.05X + 0.1U \nonumber\\
    Y_2 \sim & Bern(p_2); p_2 = 0.1 + 0.05Z + 0.05X + 0.1U 
\end{align}
where $Z$, $X$, $Y_i$ are the instrument, cause and outcome respectively and all variables are binary.
There can be three possible datasets depending on which Y is chosen as the outcome: $D_0(Z, X, Y_0)$, $D_1(Z, X, Y_1)$,$D_2(Z, X, Y_2)$. $Z$ is a valid instrument only when the outcome is $Y_0$, not for $Y_1$ and $Y_2$ because they violate the exclusion restriction. Although Pearl-Bonet test is able to rule out $D_1$ as an invalid-IV dataset, Palmer et al. find that it is inconclusive for $D_0$ and $D_2$. 

We validate the same three datasets using the NPS test  by simulating 2000 data points from each of their causal models. Table~\ref{tab:palmer-results} shows that comparing Validity-Ratio can be used to identify the datasets for which $Z$ is a valid instrument. We assume a uniform prior over models  within  valid-IV and invalid-IV model classes and use the equation from Corollary~\ref{cor:uniform-prior}. Further, in the absence of any additional information, we can assume an equal probability of the instrument being valid or invalid ($P(M_1)=P(M_2)$). The second and third columns show the log marginal likelihood for invalid-IV models when either of exclusion or as-if-random is violated. This leads to the Validity Ratio shown in the fifth column, as a ratio of marginal likelihood of the Valid-IV model class over marginal likelihood of the Invalid-IV model class. Validity Ratio is the highest (nearly 20) for $D_0$ and the lowest ($<10^{-13}$) for $D_2$, thereby clearly distinguishing between the two datasets. Dataset $D_1$ has a Validity Ratio less than 1, indicating that it is less likely to be a valid instrument, especially in comparison to dataset $D_0$.

In practice, a common goal is to select a single instrument among multiple candidate instruments. Results from the NPS test indicate that $D_0$ should be chosen; it has the highest Validity-Ratio among candidate datasets. Further, given that the Validity-Ratio is greater than 1, it is also likely to be a valid instrument. In general, to determine validity, one way would be to prespecify standard thresholds for the Validity-Ratio above which an instrument is valid, as suggested by \cite{kass1995bayes}. However, we deliberately refrain from providing standard thresholds, because the judgment for validity of an instrument will anyways depend on the prior assumed for Invalid-IV and Valid-IV model classes. We recommend instead to interpret the ratio of marginal likelihoods as a \emph{benchmark} for priors: a Validity Ratio of $20$ assuming uniform priors indicates that for $D_0$ to be an invalid-IV dataset, a researcher's prior on finding an invalid-IV dataset should be at least $20$ times as strong as the prior for valid-IV dataset. In addition, note that the Validity-Ratio may also be sensitive to the assumption of uniform priors over models within invalid-IV and valid-IV classes. We leave exploration of other priors over models for future work.

\begin{table}
    \centering
    \begin{tabular}{ccccc}
        \toprule
        & \multicolumn{3}{c}{Log Marginal Likelihood} & \\
        \cmidrule(lr){2-4}
        Dataset & Exclusion Violated & As-if-random Violated & Valid IV & Validity Ratio \\
        \midrule
        $D_0: Z, X, Y_0$   & -3080 & -3086 & \textbf{-3077} & 20.1 \\ 
        $D_1: Z, X, Y_1$ & -3168 & \textbf{-3161} & -3163 & 0.13 \\
        $D_2: Z, X, Y_2$ & \textbf{-3366} & -3367 & -3397 & $3.4$x$10^{-14}$ \\
        \bottomrule
    \end{tabular}
    \caption{Validity Ratio estimates for an example open problem proposed for testing binary instrumental variables. The NPS test can distinguish between valid-IV ($D_0$) and invalid-IV ($D_1$, $D_2$) datasets. Bold values denote the maximum marginal likelihood for each dataset.}
    \label{tab:palmer-results}
\end{table}

\begin{figure}[t]
    \centering
    \includegraphics[width=0.7\textwidth]{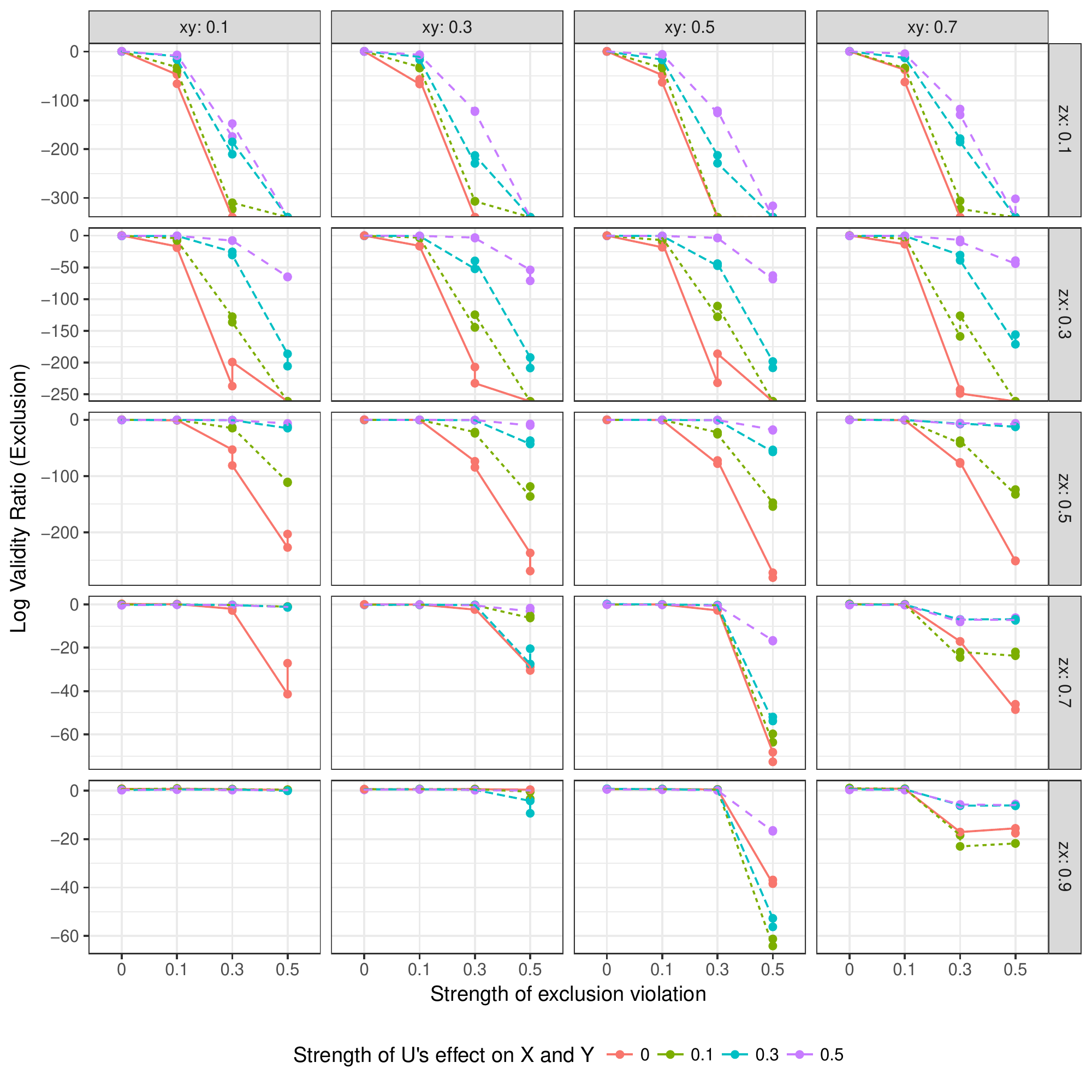}
    \caption{Log Validity-Ratio computed from the NPS test on simulated binary datasets where only exclusion is potentially violated. Strength of exclusion violation increases on the x-axis. \emph{(Rows)} zx denotes the direct effect of Z on X. \emph{(Columns)} xy denotes the direct effect of X on Y. }
    \label{fig:palmersim-excl-only}
\end{figure}
\subsection{Simulating a broad range of binary datasets}
Motivated by the example from \cite{palmer2011ivbounds}, we now construct a set of datasets that cover a broad spectrum of possible datasets with binary variables. We do so by changing the parameters of the Palmer et al.'s example model presented above. Z and U are generated from a Bernoulli distribution as before, but parameter for effect of Z on X can have five different values: $\{0.1, 0.3, 0.5, 0.7, 0.9\}$. Similarly, the effect of X on Y takes values in this set. Each of U's effect on X, U's effect on Y, U's effect on Z, and Z's effect on Y takes on values from the set $\{0, 0.1, 0.3, 0.5\}$. For simplicity, we assume that U's effect on X and Y is the same. Combined, these parameters lead to 5x5x4x4x4=1600 simulations, each of which yields a different causal model. From each causal model, we generate an i.i.d. dataset with size=50000 of  $<Z, X, Y>$ tuples.  

These simulated datasets span the range of datasets with a valid or invalid instrument. When the parameters for the effect of U on Z and the effect of Z on Y are zero, the causal model contains a valid instrument. Otherwise, it contains an invalid instrument. We make the same assumptions as before: equal prior probability of an invalid or valid instrument, and a uniform prior over causal models within both Valid-IV and Invalid-IV model classes. On each dataset, we compute the Validity-Ratio using the equation from Corollary~\ref{cor:uniform-prior}.\footnote{Unlike the NPS algorithm, here we compute the Validity-Ratio for all datasets for comparison, irrespective of whether they pass the Pearl-Bonet necessary test.}  

\subsubsection{NPS test can detect exclusion violation, except when instrument is strong}
We first look at violation of the exclusion restriction. For this case, we consider all datasets where as-if-random is satisfied (i.e., effect of $U$ on $Z$ is zero). Figure~\ref{fig:palmersim-excl-only} shows the log Validity-Ratio as the strength of the exclusion violation is increased. We find that when the parameter for effect of Z on X (\emph{instrument strength}) is below 0.5, Validity Ratio is below 1 consistently even for minor violation of the exclusion restriction. This holds true even as the true causal effect is varied: scanning horizontally through the rows shows a similar trend. In addition, the detectability of exclusion violation depends on the effect of confounders U on X and Y. When confounders do not affect X and Y (red line), Validity Ratio is the most sensitive to violations of exclusion. As the effect of confounders increases, Validity Ratio becomes less sensitive in detecting exclusion violation. In addition to detecting violations, the NPS test also correctly returns a Validity Ratio close to 1 or higher for valid instruments. The mean Validity-Ratio across all valid instruments is 1.8, with a maximum value of 12. There are some valid-IV models for which the Validity-Ratio falls below 1, but in all cases it is higher than 0.3 (or roughly, -1 on the log scale). 

Above results indicate that exclusion can be tested by inspecting the Validity-Ratio as long as the instrument is not too strong (effect parameters $<0.5$).  Further, the NPS test is able to identify violation of the exclusion restriction best when there is little confounding between X and Y. However, this is still a secondary effect and even at a confounding effect of U on X and Y at 0.5, the NPS test can successfully identify instruments that violate exclusion.

\begin{figure}[t]
    \centering
    \includegraphics[width=0.7\textwidth]{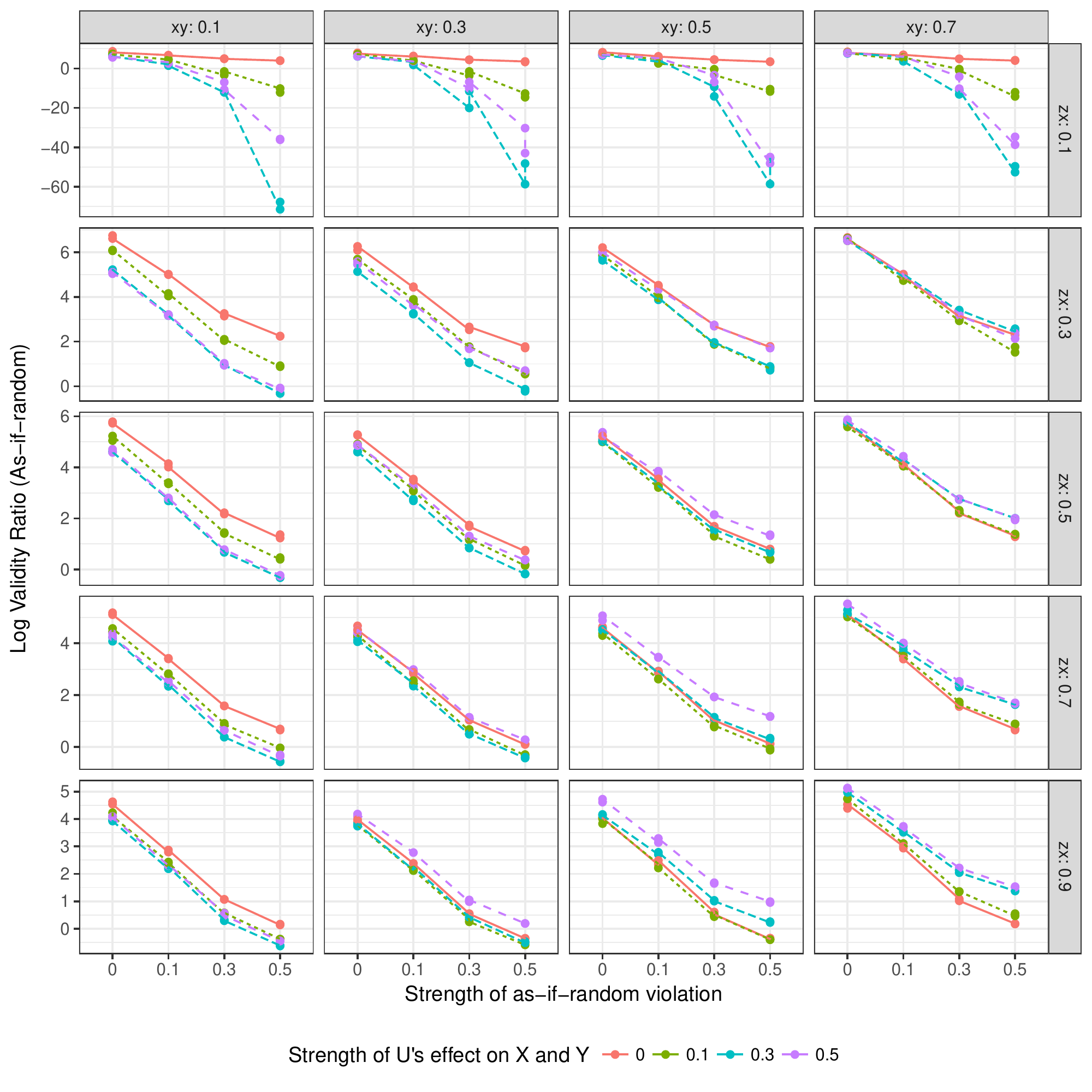}
    \caption{Log Validity-Ratio computed from the NPS test on simulated binary datasets where only as-if-random is potentially violated. Strength of as-if-random violation increases on the x-axis. \emph{(Rows)} zx denotes the direct effect of Z on X. \emph{(Columns)} xy denotes the direct effect of X on Y. }
    \label{fig:palmersim-air-only}
\end{figure}

\subsubsection{As-if-random is hard to detect, except when instrument is very weak}
Next, we look at violation of the as-if-random restriction only. For this case, we consider all datasets where the exclusion condition is satisfied (i.e., effect of Z on Y is zero). The obtained log Validity-Ratio as the parameter for effect of U on Z is varied is shown in Figure~\ref{fig:palmersim-air-only}. We find that violation of the as-if-random is harder to detect than exclusion. When the instrument is very weak (effect parameter for Z on X is 0.1), the Validity-Ratio goes below 1 as the strength of as-if-random violation is increased. This result is consistent even as the direct causal effect from X to Y is varied. Interestingly, the detection of a violation is better when there is a strong confounding effect of U on X and Y, and worse when the confounding effect is near zero. We conjecture that this is because the test is trying to detect an effect of U on Z and it is helpful if U also has a non-trivial on X and Y to gain comparison from. As for the exclusion assumption above, Validity-Ratio for datasets with a valid instrument is close to or higher than 1, indicating that they are probably valid. 

However, in the case where instrument strength increases to 0.3 and above, the Validity Ratio stays above 1 even when as-if-random is violated. As we found in Section~\ref{sec:uniform-asifrandom-res}, these results indicate that the NPS test is unable to detect violations of the as-if-random assumption. As in that section, it is also likely that violations of as-if-random assumption lead to lesser bias in causal estimation than the exclusion assumption.

\begin{figure}[t]
    \centering
    \includegraphics[width=0.7\textwidth]{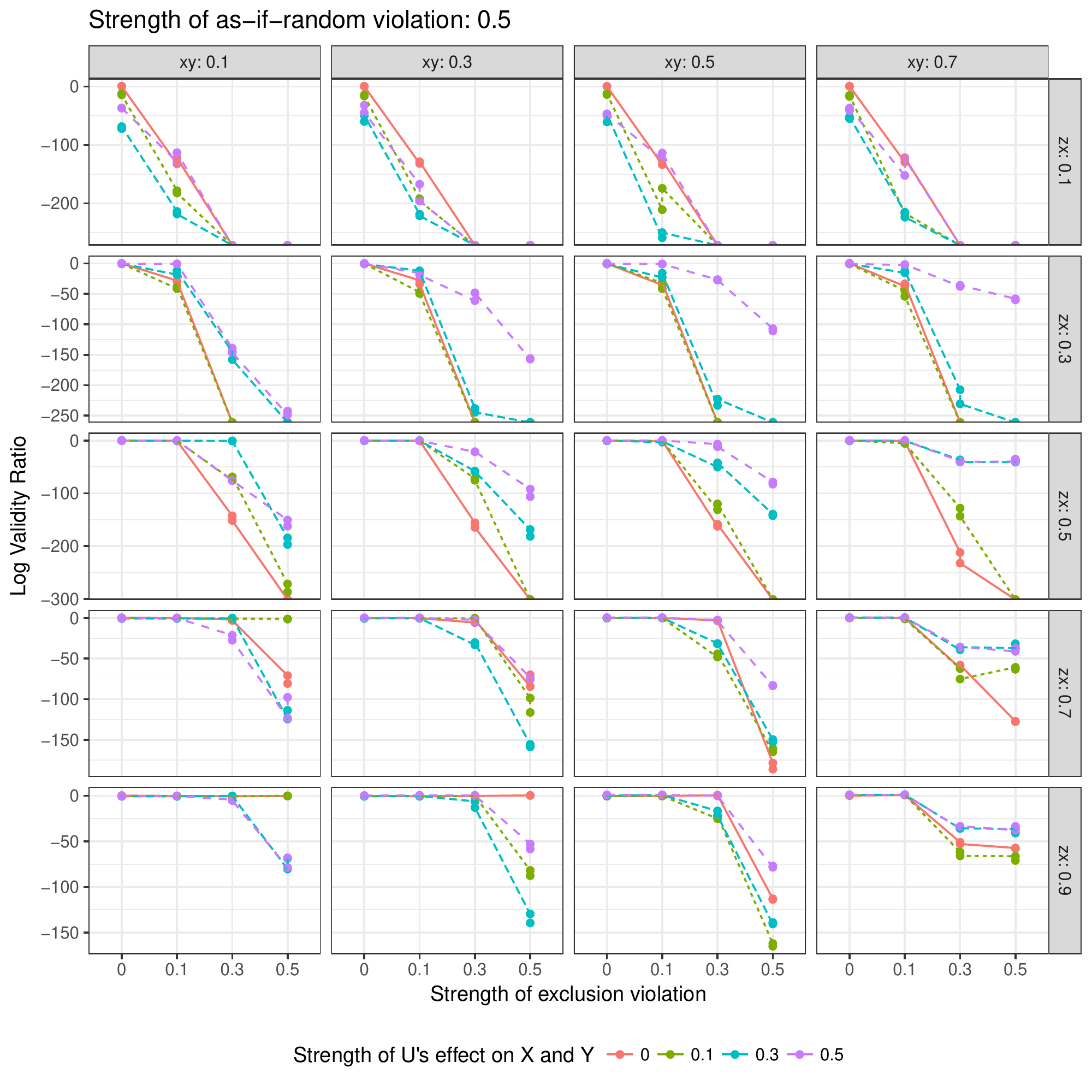}
    \caption{Log Validity-Ratio computed from the NPS test on simulated binary datasets where both exclusion and as-if-random are potentially violated. Strength of as-if-random violation is fixed at 0.5. \emph{(Rows)} zx denotes the direct effect of Z on X. \emph{(Columns)} xy denotes the direct effect of X on Y. }
    \label{fig:palmersim-excl-air}
\end{figure}
\subsubsection{Violations of both assumptions is easier to detect}
Finally, we look at the case when both exclusion and as-if-random are violated. Here we considered all simulated datasets. Figure~\ref{fig:palmersim-excl-air} shows the log Validity-Ratio as the strength of exclusion violation varies, for a fixed as-if-random violation of 0.5 (i.e., the parameter for U's effect on Z is 0.5). When both exclusion and as-if-random are violated, it becomes easier to identify datasets with invalid instruments. Even when the is instrument is moderately strong (effect of Z on X is 0.7), we find that Validity Ratio quickly drops to less than 1 as the strength of exclusion violation increases. This pattern is consistent as the true causal effect of X on Y is varied across datasets. When the instrument's effect on X is the strongest at 0.9, NPS test can still detect violations of exclusion with a severity higher than 0.3.  
Detection of invalid instruments becomes weaker as the strength of as-if-random violation is decreased. We include results for other values of the as-if-random violation (i.e., effect of U on Z) in Appendix C. 

\subsubsection{Summary of results}
Overall, these results show that the NPS test can detect violations of exclusion, but violations of as-if-random are not detected as reliably. Further, if both exclusion and as-if-random are violated, the NPS test can detect them more easily. In all cases, distinguishing between an invalid and valid instrument is more reliable when the instrument's effect on the cause is weak, as is the case with many observational studies. 
Finally, for all datasets that contained a valid instrument, the NPS test correctly returns a Validity Ratio close to or higher than 1, thus indicating the probable validity of the instrument.

\section{USING NPS TEST TO VALIDATE PAST IV STUDIES}
\label{sec:applications}
In this section we use the NPS test to evaluate empirical studies based on instrumental variables. We select two  seminal and highly cited studies on instrumental variables and a sample of recent studies from a leading economics journal, \emph{American Economic Review}.

\begin{figure}[t]
\subfloat[Unconditional IV]{\includegraphics[scale=0.45]{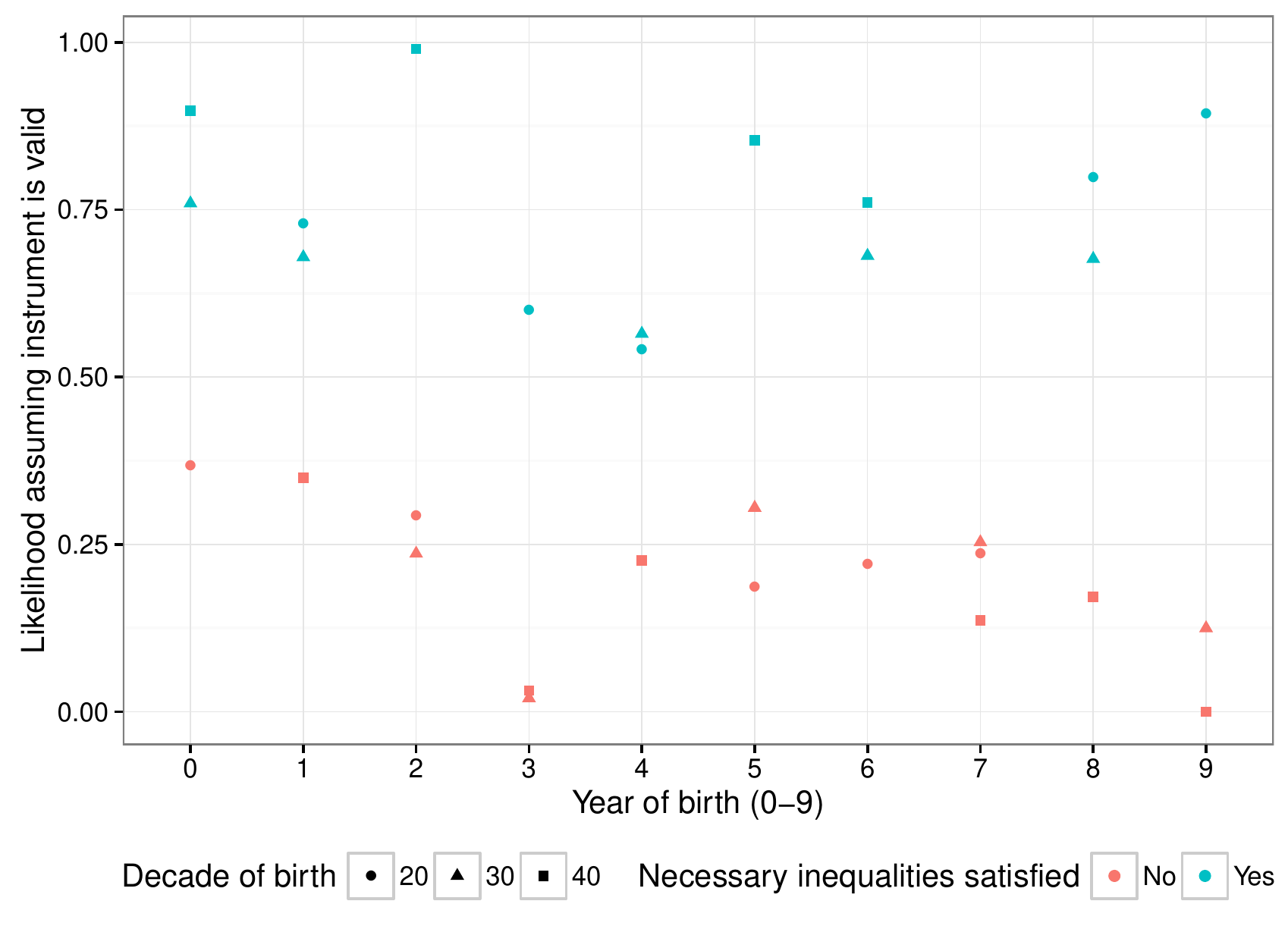} \label{fig:angrist-unconditional-iv}} %
    \qquad
    \subfloat[Conditional IV]{\includegraphics[scale=0.45]{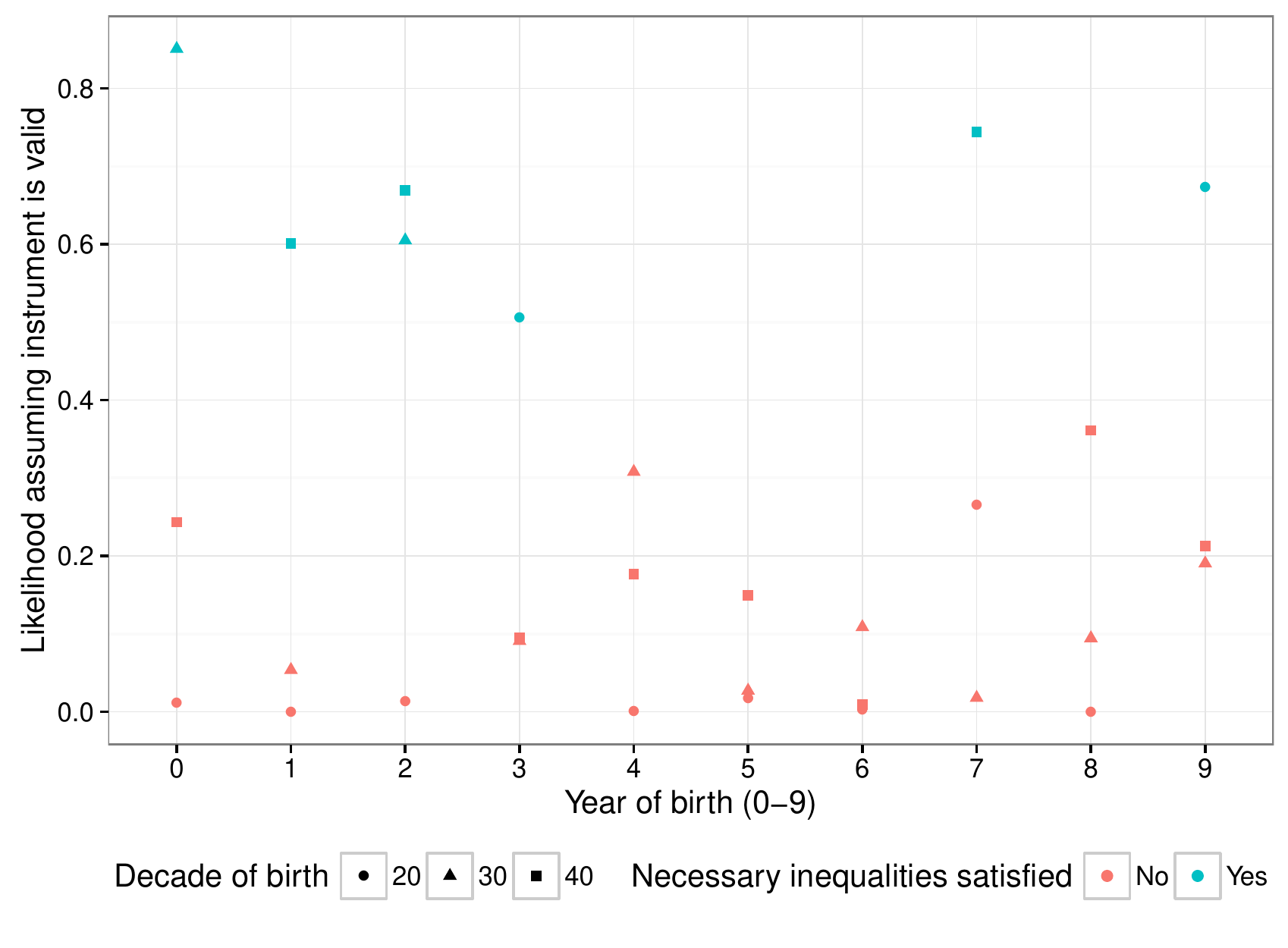} \label{fig:angrist-conditional-iv} }
    \caption{p-values from the Pearl-Bonet necessary test for instruments used in a past IV study. Left panel shows unconditional instruments, while the right panel instruments conditioned on relevant covariates. Assuming a p-value threshold of $0.05$, many of the instruments are determined to be invalid by the  test.} 
\end{figure}

\subsection{Seminal IV studies on the effect of schooling on wages}
Returns of schooling on future income was one of the first applications that instrumental variable studies were applied to. We apply the NPS test to two of the seminal instrumental variable studies \citep{angrist1991,card1993}. These studies are both highly cited, yet concerns about their validity  continue until today \citep{bound1995-ivproblems,buckles2013season}. We show how the NPS test can provide evidence and help evaluate the instruments used in these studies. 

\subsubsection{Effect of compulsory schooling on future wages}
The first paper estimates the effect of compulsory schooling on future earnings of students \cite{angrist1991}. In the original analysis, $Z$ is the quarter of birth, which was binarized to indicate that the student was either born in the first quarter or the last three quarters. $X$ is years of schooling and $Y$ is the (log) weekly earnings of individuals. Using yearly cohorts, the authors define a separate instrumental variable for each year of birth. The causal effect of interest is the effect of years of schooling on future earnings.\footnote{Dataset available at \url{http://economics.mit.edu/faculty/angrist/data1/data/angkru1991}}
Years of schooling and earnings are both reported as continuous numbers, in a bounded range. We follow the simplest discretization by binarizing both these variables at their mean. To check robustness against outliers, we also tried using median as the cutoff and obtained similar results. 

As in the original paper, we first use the IV method without conditioning on any covariates. Applying the NPS Algorithm shows that nearly half of the instruments do not satisfy the inequalities of the necessary test. However, some of these may be due to sampling variability. At a significance level of $\alpha=0.05$,   3 yearly instruments do not pass the necessary test, as Figure~\ref{fig:angrist-unconditional-iv} shows. This indicates that one or more of the three assumptions---monotonicity, exclusion and as-if-random---is violated, at least when all variables are binary. The null hypothesis in the necessary test is that an instrument is valid. 

In practice, however, unconditional instrumental variables are rare. The original paper proceeds to construct conditional instrumental variables based on related covariates in the dataset. We replicate conditioning on covariates by implementing a ``partialling out'' technique \citep{baum2007partialout}. Based on the Frisch-Waugh-Lovell theorem, the partialling out technique provides the equivalent effect of conditioning on covariates, provided the underlying causal model is linear (which the authors assume). Figure~\ref{fig:angrist-conditional-iv} shows the results of the necessary test when all instruments are conditioned on covariates. Contrary to intuition, a bigger fraction of conditional instruments are now invalid at the $5\%$ significance level using the necessary test. Combined, our results on unconditional and conditional IVs suggest that many of the instruments used in the original analysis may not be valid, at least under the transformation when treatment and outcome are binarized. Note that these results do not necessarily invalidate the analysis in the paper: the instruments in the original data may still be valid when X and Y are continuous, as binarizing the variables can also lead to violation of one of the assumptions. 

We also estimate the Validity Ratio for the instruments that pass the necessary test. We find that more than one quarter of the instruments have a Validity-Ratio lower than $0.001$, with the minimum and maximum validity ratio being $1.1$x$10^{-7}$ and $0.5$ respectively. Thus, while some of the instruments may be probably valid, the NPS test indicates that many others are likely to be invalid. In this example study, the NPS test can be used to select the top-ranked probably valid instruments for analysis, and filter out the rest.  

\subsubsection{Return of college education on future earnings}
Another early study in the instrumental variables literature was on the effect of college education on future earnings of students \citep{card1993}. This study used a person's distance from college as an instrument to estimate the causal effect of college education.
We follow a similar protocol as above. $Z$ is distance from college, which was already binarized in the original study. $X$ is years of education and $Y$ is the log weekly wages.\footnote{Dataset available at \url{http://davidcard.berkeley.edu/data_sets.html}} 
After binarizing $X$ and $Y$, we find that data from Card's study does not pass the necessary test for IV validity. This is corroborated by estimating the Validity-Ratio for the same data. For both mean and median split for binarization, the Validity Ratio is lower than 0.003.

However, in their analysis, Card also considers a conditional instrumental variable that conditions on a number of secondary variables, such as race and geography. When we use the partialling out technique from above to condition on these variables and rerun the NPS test, we find that the dataset passes the necessary test and yields a Validity Ratio of 0.2, indicating a higher likelihood of being valid than the unconditional instrument.

\begin{table}
    \centering
    \begin{tabular}{lrlll}
	\toprule
        Study name & Num. Observations & IV Strength & Pearl-Bonet Test & Validity Ratio \\
	\midrule
	\textbf{Randomized Experiment} &&&& \\
	\multicolumn{5}{l}{\emph{National Job Training Partnership Act (JTPA) Study \citep{abadie2002jtpa}}} \\
	& 5102 & 0.58 & Pass & 3.4 \\
	
	\textbf{Instrumental Variable Studies} &&&& \\
	\multicolumn{5}{l}{\emph{Effect of rural electrification on employment in South Africa (2011) \citep{dinkelman2011southafrica}}} \\
	--Type 0 & 1816 & 0.1 & Pass & 3.6\\
	--Type 1 & 1816 & 0.1 & Fail (p=0.26) & 0.002 \\
	--Type 2 & 1816 & 0.16 & Fail (p=0.16) & 0.0009\\
	--Type 3 & 1816 & 0.05 & Fail (p=0.11) &  0.001\\

	\multicolumn{5}{l}{\emph{Effect of Chinese import competition on local labor markets (2013) \citep{autor2013china}}}  \\
        --Outcome(population change) & 1444 & 0.59 & Pass & 0.3\\
        --Outcome(employment) & 1444 & 0.59 & Pass & 0.3\\

	\multicolumn{5}{l}{\emph{Effect of credit supply on housing prices (2015) \citep{favara2015credithousing}}} \\
        --Outcome(nloans)  & 11107 & -0.009 &  Fail (p=0.003) & 0.011 \\
        --Outcome(vloans) & 11107 & -0.003 & Fail (p=0.005) & 0.006  \\
        --Outcome(lir) & 11107 & -0.01 & Fail (p=0.004) &  0.0004 \\

	\multicolumn{5}{l}{\emph{Effect of subsidy manipulation on Medicare premiums (2015) \cite{decarolis2015medicare}}} \\ 
        --Unconditioned& 170 & 0.60 & Pass & 1.02  \\
        --Conditioned & 170 & 0.42 & Pass & 0.04 \\

	\multicolumn{5}{l}{\emph{Effect of Mexican immigration on crime in United States (2015) \citep{chalfin2015mexico}}} \\
        --Unconditioned & 182 &0.50 & Pass & 0.07 \\
        --Conditioned & 182 & 0.22 & Pass & 0.005 \\
	\bottomrule
    \end{tabular}
    \caption{Results of the NPS test on recent studies from the American Economic Review. Many of the instruments fail the necessary test and have comparatively low Validity-Ratios, indicating that they may be invalid, at least under the binary transformation.}
    \label{tab:recent-aer-studies}
\end{table}

\subsection{Recent IV studies in the American Economic Review}
We now apply the NPS test to validate more recent IV studies. To select recent studies, we searched for papers published in the American Economic Review from 2011-2015 that had `instrumental variable" or ``instrument" mentioned in their title or abstract. From this set, we filtered out studies that did not provide full datasets for replication, leaving us with five studies on the causal effect of diverse economic treatments such as rural electrification \citep{dinkelman2011southafrica}, credit supply \citep{favara2015credithousing}, subsidy manipulation \citep{decarolis2015medicare}, foreign import competition \citep{autor2013china}, and foreign immigration \citep{chalfin2015mexico}. As a comparison benchmark, we also include an instrumental variable study based on data from a randomized experiment \citep{abadie2002jtpa}, which almost surely should pass the NPS test. 

For each of the studies, we use code provided by the authors to construct a dataset of three variables $(Z, X, Y)$, where Z is the instrument. If the authors condition on covariates, then we use the partialling out technique to process the $(Z, X, Y)$ dataset. Finally, we binarize each variable at its mean, unless it is already binarized.  Table~\ref{tab:recent-aer-studies} presents results from the NPS test. First, we find that the randomized experiment passes the Pearl-Bonet test and obtains a Validity Ratio of 3.4, thus providing evidence for a probably valid instrument. In contrast, 2 out of 5 studies do not pass Pearl-Bonet test when binarized. They also report significantly low estimates of the Validity Ratio. For the other three studies, the Validity Ratio indicates a measure of their probable validity. 
As an example, the conditional instrument for the study on Mexican immigration obtains a Validity Ratio of 0.005, providing evidence for the invalidity of the instrument (i.e.,  the data does not support validity of the instruments used, at least under the transformation of binary variables).  In the remaining two studies the Validity Ratio is close to 1 so we cannot reject the validity of the instrument using the NPS test, thereby proving inconclusive about their validity.

\section{DISCUSSION AND FUTURE WORK}
\label{sec:discussion}


We presented a probably sufficient test for instrumental variables using necessary tests proposed by past work. Simulation results show that the test is more effective for detecting violation of the exclusion assumption, and that  effectiveness of the test increases as the strength of the instrument decreases. Therefore, while the NPS test cannot always verify whether an instrument is valid, it is more effective when the instrument is weak. Fortunately, many observational studies are based on instruments with low $Z$-$X$ correlation, where NPS can be applied. 

Nevertheless, the proposed test has several limitations. First, it relies on the specification of a prior over causal models for both the Valid-IV and Invalid-IV model classes. In this paper we chose a uniform prior but it is possible to choose other priors. If sufficient data is available, a possible method is to split the sample into two and use the first sample to estimate relative likelihoods of different models~\citep{ohagan1995fractional}. This estimated likelihood can then be used as the prior over models in the estimation phase. It will be useful to study the sensitivity of the Validity-Ratio to changes in this prior, which we leave for future work. 
Second, the proposed implementation of the NPS test works only for discrete variables. Extensions to continuous variables can increase the applicability of this test. Third, even for discrete variables, the test is often inconclusive. If the Validity-Ratio lies close to 1 (from -1 to 0 on the log scale), then we are unable to distinguish between valid and invalid instruments. Based on the simulation results, we conjecture that in such cases the resultant causal estimate will not have high bias even for invalid instruments, but this claim needs more evidence. In addition, we would also like to study if there is a natural threshold that can be set for identifying valid instruments in such cases. 

More generally, the NPS test is an example of a general Bayesian testing framework for causal models. 
Looking forward, the proposed test can be used to compare potential instruments for their validity, allow transparent comparisons between multiple IV studies, and enable a data-driven search for natural experiments~\citep{sharma2016-splitdoor}.  


\acks{We acknowledge Jake Hofman and Duncan Watts for their valuable feedback throughout the course of this work. We also thank Miro Dudik, Akshay Krishnamurthy, Justin Rao, Vasilis Syrgkanis and Michael Zhao for helpful suggestions.
}

\appendix
\section*{Appendix A} \label{app:derive-ratio}
Details on computing the Validity Ratio. 
\subsection*{Calculating the numerator}
When both conditions are satisfied, the numerator of Equation~\ref{eqn:implement-ratio-valid-iv} can be written as:
\begin{align} \label{eqn:implement-numerator-valid-iv}
	P(D|\theta) &= \prod_{j=1}^{Q} (\sum_{r_{zxy}='000'}^{'133'} P(R_{Z} = r_{Z})(P(Z=z_j| \theta, r_{z}) P(R_{XY}=r_{xy})(P(X=x_j, Y=y_j | \theta, r_{zxy} ))^{Q_j}  \nonumber \\
	       &=\prod_{j=1}^{Q} (\sum_{r_{z}='0'}^{'1'} P(R_{Z} = r_{Z})(P(Z=z_j| \theta, r_{z}))^{Q_j} (\sum_{r_{xy}='00'}^{'33'}P(R_{XY}=r_{xy})(P(X=x_j, Y=y_j | \theta, r_{xy} ))^{Q_j} 
\end{align}

Note that for a fixed value of $R_Z$, $Z$ can be uniquely determined. Similarly, for a given value of $Z$ and $R_{XY}$, $X$ and $Y$ can be deterministically evaluated. Thus, the above expression reduces to:
\begin{align}
	P(D| \theta) &=\prod_{j=1}^{Q} (\sum_{r_{z}='0'}^{'1'} \theta_{r_z})^{Q_j} (\sum_{r_{xy}='00'}^{'33'} \theta_{r_{xy}} )^{Q_j}  \nonumber \\ 
			  	                  &= \theta_{rz=0}^{Q_0} (\theta_{rxy=00}+\theta_{rxy=20}+\theta_{rxy=02}+\theta_{rxy=22})^{Q_0} \nonumber \\
				   & \text{\ \ \ \ \ }\theta_{rz=0}^{Q_1} (\theta_{rxy=01}+\theta_{rxy=21}+\theta_{rxy=03}+\theta_{rxy=23})^{Q_1} \nonumber  \\
				   & \text{\ \ \ \ \ }\theta_{rz=0}^{Q_2} (\theta_{rxy=11}+\theta_{rxy=10}+\theta_{rxy=31}+\theta_{rxy=30})^{Q_2} \nonumber  \\
				   & \text{\ \ \ \ \ }\theta_{rz=0}^{Q_3} (\theta_{rxy=12}+\theta_{rxy=13}+\theta_{rxy=32}+\theta_{rxy=33})^{Q_3} \nonumber  \\
				   & \text{\ \ \ \ \ }\theta_{rz=1}^{Q_4} (\theta_{rxy=00}+\theta_{rxy=02}+\theta_{rxy=10}+\theta_{rxy=12})^{Q_4}  \nonumber \\
				   & \text{\ \ \ \ \ }\theta_{rz=1}^{Q_5} (\theta_{rxy=01}+\theta_{rxy=03}+\theta_{rxy=11}+\theta_{rxy=13})^{Q_5} \nonumber  \\
				   & \text{\ \ \ \ \ }\theta_{rz=1}^{Q_6} (\theta_{rxy=20}+\theta_{rxy=30}+\theta_{rxy=21}+\theta_{rxy=31})^{Q_6} \nonumber  \\
				   & \text{\ \ \ \ \ }\theta_{rz=1}^{Q_7} (\theta_{rxy=22}+\theta_{rxy=32}+\theta_{rxy=23}+\theta_{rxy=33})^{Q_7}
\end{align}

The above equation leads to the following simplification for the numerator of Equation~\ref{eqn:implement-ratio-valid-iv}.
\begin{align}
	\int_{M1:m \text{ is valid}} P(D|m) dm &= \iint_{\theta_{rz},\theta_{rxy}} \prod_{j=1}^{Q} \theta_{rz=z}^{Q_j} (\theta_{rxy=a} + \theta_{rxy=b}+\theta_{rxy=c} + \theta_{rxy=d})^{Q_j} d\theta_{rz} d\theta_{rxy} \nonumber \\
	  &= \int \prod_{j=1}^{Q} \theta_{rz=z}^{Q_j} d\theta_{rz} \int  \prod_{j=1}^{Q} (\theta_{rxy=a} + \theta_{rxy=b}+\theta_{rxy=c} + \theta_{rxy=d})^{Q_j} d\theta_{rxy} 
\end{align}

The above integral has a form equivalent to the hyperdirichlet integral~\citep{hankin2010hyperdirichlet}, for which no tractable closed form exists except in a few special cases.\footnote{We say \textit{tractable} because it is possible to decompose the hyperdirichlet integral into a sum of exponentially many dirichlet integrals \citep{cooper1992bayesian}, but that will not be computationally feasible. } We therefore resort to approximate methods for estimating the integral. For binary X, Y and Z, the maximum dimension of the integral will be 16, so we recommend using approximate integral techniques over the unit simplex. \citep{genz2003cubature} For discrete variables, monte carlo methods for estimating marginal likelihood,  such as annealed importance sampling \citep{neal2001annealed} or nested sampling \citep{skilling2006nested,feroz2009multinest} may be more appropriate.

\subsection*{Calculating the denominator}
We can calculate the denominator of Equation~\ref{eqn:implement-ratio-valid-iv} in a similar way as the numerator, except that the exact integral expression will vary based on the extent of violation of as-if-random and exclusion restrictions. 

\subsubsection*{When Exclusion is violated}
Following Equations~\ref{eqn:implement-ratio-valid-iv}, the denominator can be expressed similarly to \ref{eqn:implement-numerator-valid-iv}. The only difference is that $\theta_{rxy}$ will be 4x16=$64$-dimensional integral.
\begin{align} \label{eqn:ml-exclusion-violated}
	P(D|\theta) 
	       &=\prod_{j=1}^{Q} (\sum_{r_{z}='0'}^{'1'} P(R_{Z} = r_{Z})(P(Z=z_j| \theta, r_{z}))^{Q_j} (\sum_{r_{xy}='0,0'}^{'3,15'}P(R_{XY}=r_{xy})(P(X=x_j, Y=y_j | \theta, r_{xy} ))^{Q_j} 
\end{align}

However, the above formulation corresponds to a full violation of the exclusion condition, assuming all $\theta_{ry}\in \{0,1,2,3,...,15\} \setminus \{0,3,12,15\}$ are non-zero. In practice, exclusion can be violated even if a single $R_Y$ in that set is non-zero.
Therefore, a stronger and realistic way of estimating marginal likelihood under an invalid-IV is to compute the maximum of all marginal likelihoods for causal models where one of the $R_Y$ corresponding to an exclusion violation is non-zero. This would mean computing 12 integrations over 4x5=20 dimensions, one for each for nonzero $R_Y$ that results in an exclusion violation.

\subsubsection*{When As-if-random is violated}
Fortunately, when as-if-random condition is violated, we can obtain a closed form solution for the integral. Proceeding from Equation~\ref{eqn:implement-general-ml}, we cannot simplify the marginal likelihood as a product of two independent integrals and thus obtain: 
\begin{align} \label{eqn:implement-air-violated}
	\int P(D|\theta) d\theta &=
	\int_{M2} \prod_{j=1}^{Q} (\sum_{r_{zxy}=0,0,0}^{1,3,3} P(R_{XYZ} = r_{xyz}) P(Z=z_j, X=x_j, Y=y_j | \theta, r_{zxy} ))^{Q_j} 
\end{align}

This, however, means that each $\theta_{rzxy}$ occurs exactly once in the integral, allowing a transformation of the integral to a dirichlet integral. 
The closed form integral is given by,
\begin{align} \label{eqn:ml-air-violated}
\int P(D|\theta) d\theta &= \frac{\prod_{j=1}^{Q} \Gamma(4+Q_j)}{(\Gamma(4))^Q \Gamma(\sum_{j=1}^{Q} \Gamma(4+Q_j)))}
\end{align}
where $\Gamma(n)=factorial(n-1)$ is the Gamma function.

\subsubsection*{When both are violated}
When both exclusion and as-if-random conditions are violated, we can again obtain a closed form solution. The integral expression is similar to that for as-if-random violation (Equation~\ref{eqn:implement-air-violated}), except that the number of dimensions of theta increases to 2x4x16=128. The denominator of the Validity Ratio can be evaluated as:

\begin{align} \label{eqn:ml-both-violated}
	\int P(D|\theta) d\theta &=
	\int_{M2} \prod_{j=1}^{Q} (\sum_{r_{zxy}=0,0,0}^{1,3,15} P(R_{XYZ} = r_{xyz}) P(Z=z_j, X=x_j, Y=y_j | \theta, r_{zxy} ))^{Q_j} \nonumber \\
	&= \frac{\prod_{j=1}^{Q} \Gamma(16+Q_j)}{(\Gamma(16))^Q \Gamma(\sum_{j=1}^{Q} \Gamma(16+Q_j)))}
\end{align}

\section*{Appendix B}\label{app:monotonicity-proof}
\begin{theorem} 
For any data distribution $P(X, Y, Z)$ generated from a  valid-IV model that also satisfies monotonicity, the following inequalities hold:
\begin{align} \label{eqn:discrete-monotonicity-thm}
P(Y=y, X \geq x| Z=z_0) \leq P(Y=y, X \geq x| Z=z_1) \quad ... \quad \leq P(Y=y, X \geq x| Z=z_{l-1}) \quad \forall x,y \nonumber\\
P(Y=y, X \leq x| Z=z_0) \geq P(Y=y, X \leq x| Z=z_1) \quad ... \quad \geq P(Y=y, X \leq x| Z=z_{l-1}) \quad \forall x,y
\end{align}
where $Z$, $X$ and $Y$ are ordered discrete variables of levels $l$, $n$ and $m$ respectively and $z_0 \leq z_1 ... \leq z_{l-1}$.
\end{theorem}

\begin{proof}
Consider the first set of inequalities with $P(Y=y, X\geq x|Z=z_k)$, for some $X=x$ and $Y=y$. Based on the structure of a Valid-IV causal model (Figure~\ref{fig:std-iv-model}), we can factorize $P(Y, X|Z)$ as:
\begin{align*}
P(Y=y, X \geq x|Z=z) = P(X \geq x|Z=z) P(Y=y|X=x,Z=z) = P(X \geq x|Z=z) P(Y=y|X \geq x)
\end{align*}
$P(Y=y|X \geq x)$ is independent of $Z$. Therefore, as $Z$ varies, $P(Y=y, X \geq x|Z=z)$ only depends on $P(X \geq x, Z=z)$.

Using the structural equations for $x=g(z, u)$ from Equation~\ref{eqn:validiv-structural-x}, we obtain for any $x$ and $z \in \{z_0, z_1, ...z_{l-1}\}$:
\begin{align} \label{eqn:discrete-mono-proof-1}
P(X \geq x| Z=z_k) = P(R_X : g(z_k,u) \geq x) 
\end{align}

By monotonicity, we know that $g(z_{k2},u) \geq g(z_{k1},u)$ whenever $z_{k2} \geq z_{k1}$. Thus, we can write:
\begin{equation} \label{eqn:discrete-mono-proof-2}
g(z_{k1},u) \geq x \Rightarrow g(z_{k2},u) \geq x \qquad \text{if } z_{k2} \geq z_{k1}
\end{equation}

Combining Equations~\ref{eqn:discrete-mono-proof-1} and \ref{eqn:discrete-mono-proof-2}, for any $k$, we can argue that the set of response variables $r_x$ that satisfy $g(z_k,u) \geq x$ will always be smaller than the set of response variables that satisfy  $g(z_{k+1},u) \geq x$. Therefore, we obtain the following inequality:
\begin{align*}
P(X \geq x| Z=z_k) = P(R_X : g(z_k,u) \geq x) \leq P(R_X : g(z_{k+1},u) \geq x) = P(X \geq x| Z=z_{k+1})
\end{align*}

Iterating over $k \in \{0,1..., l-1\}$ will provide us the first set of inequalities stated in the Theorem. We can follow a similar reasoning to derive the second set of inequalities with $P(Y=y, X\leq x|Z=z_k)$.

\end{proof}

\section*{Appendix C}\label{app:extra-sim-plot}
Here we provide two additional simulation plots for different strengths of violation of the as-if-random condition (0.1 and 0.3). 
\begin{figure}[h!]
    \centering
    \includegraphics[width=0.7\textwidth]{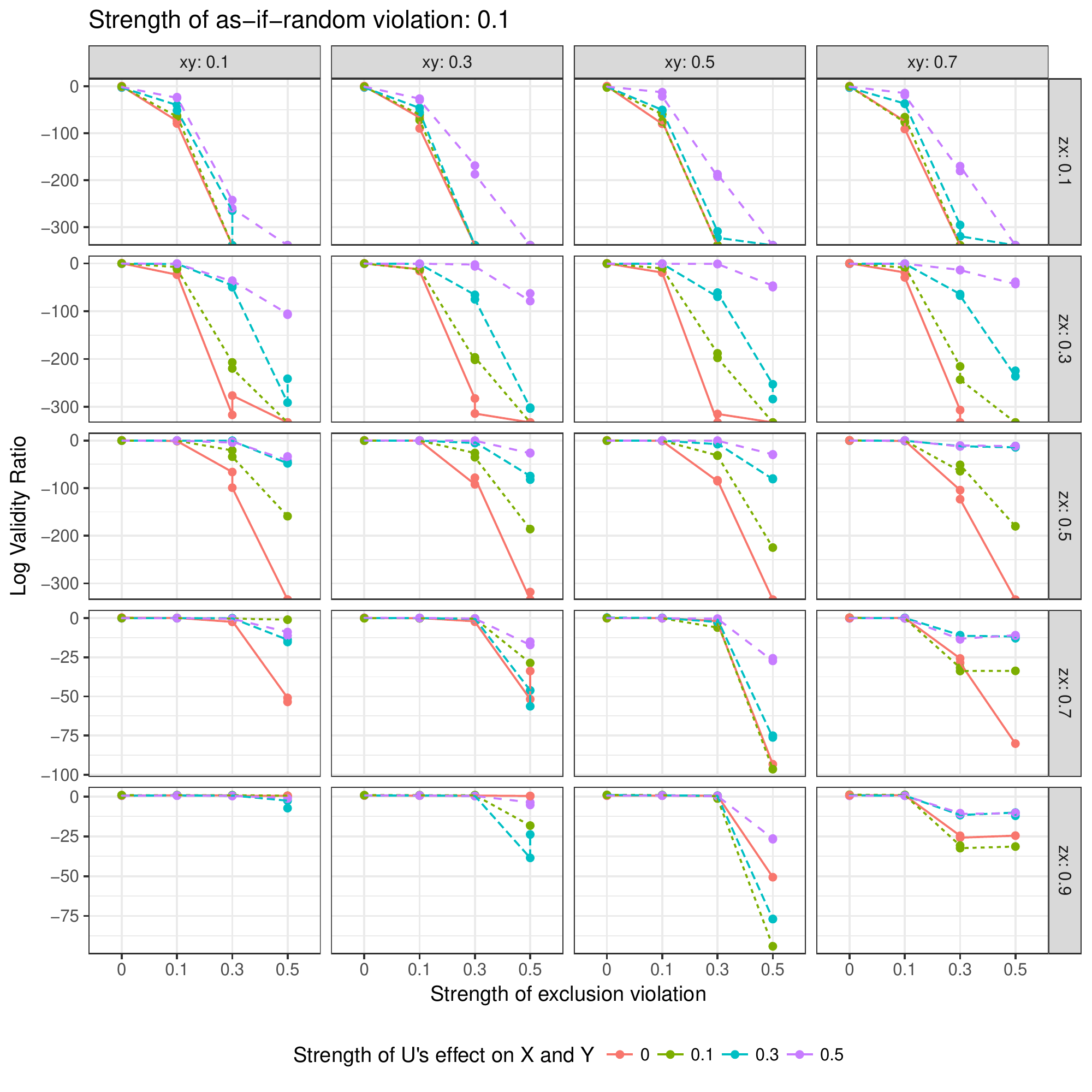}
    \caption{Log Validity-Ratio computed from the NPS test on simulated binary datasets where both exclusion and as-if-random are potentially violated. Strength of as-if-random violation is fixed at 0.1 \emph{(Rows)} zx denotes the direct effect of Z on X. \emph{(Columns)} xy denotes the direct effect of X on Y. }
    \label{fig:palmersim-excl-air}
\end{figure}
\begin{figure}[h!]]
    \centering
    \includegraphics[width=0.7\textwidth]{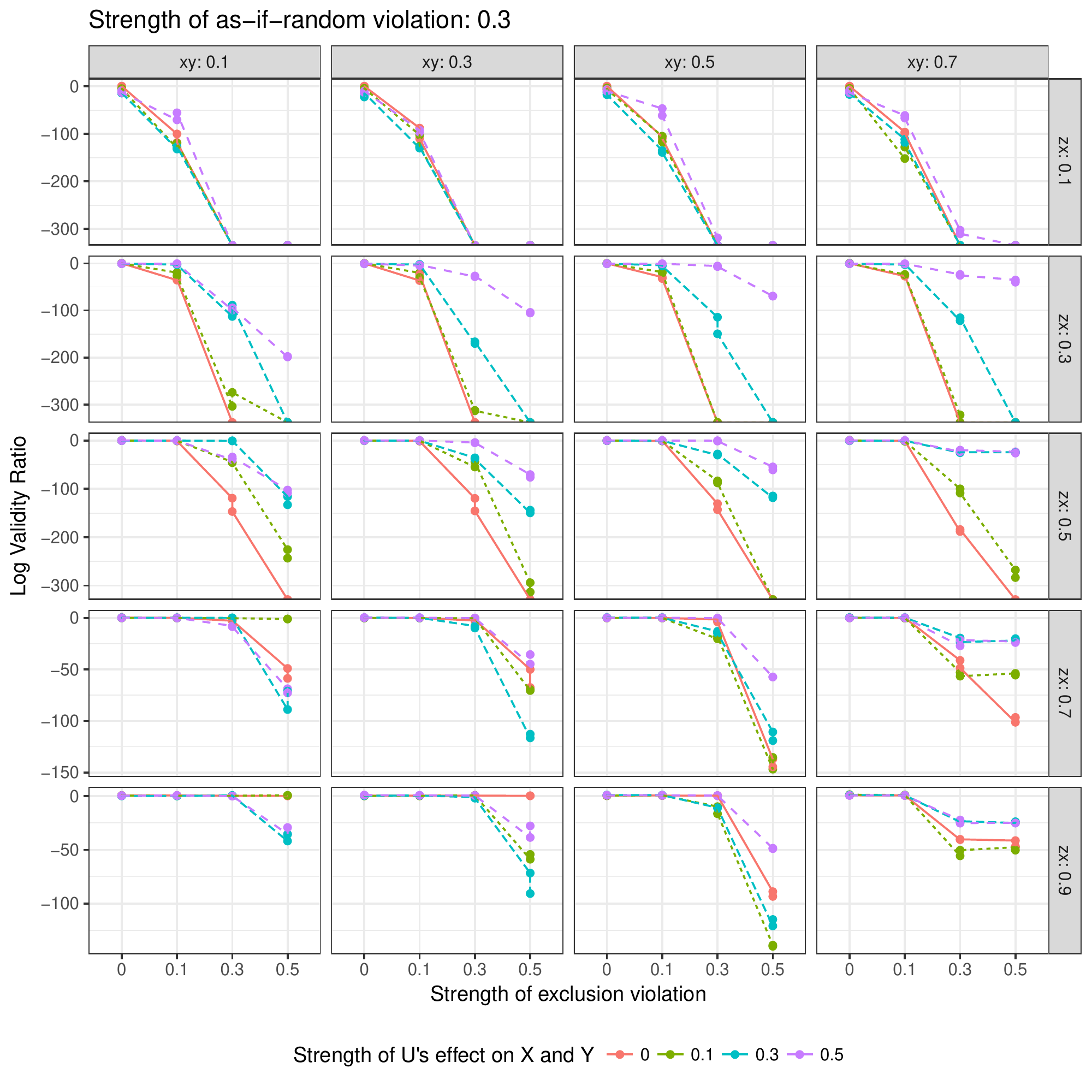}
    \caption{Log Validity-Ratio computed from the NPS test on simulated binary datasets where both exclusion and as-if-random are potentially violated. Strength of as-if-random violation is fixed at 0.3. \emph{(Rows)} zx denotes the direct effect of Z on X. \emph{(Columns)} xy denotes the direct effect of X on Y. }
    \label{fig:palmersim-excl-air}
\end{figure}

\newpage
\vskip 0.2in
\bibliography{iv-papers}
\end{document}